\documentclass[12pt, epsf, amssymb, graphicx, multirow, amsmath, color]{article}

\def\bfm#1{\mbox{\boldmath$#1$}}
\usepackage{amsmath,amssymb,amsfonts,amssymb,amsthm,fancyheadings}
\usepackage[numbers]{natbib}
\usepackage[colorlinks=true,allcolors=blue, bookmarks=false]{hyperref}
\usepackage{enumerate}
\usepackage[noend]{algpseudocode}
\usepackage{algorithmicx, algorithm}
\usepackage{MnSymbol}
\usepackage{graphicx}

\begin{document}

\voffset -0.5truecm \hoffset -1.5truecm


\renewcommand{\theequation}{\thesection.\arabic {equation}}

\newtheorem{proposition}{Proposition}
\newtheorem{corollary}{Corollary}
\newtheorem{definition}{Definition}
\newtheorem{example}{Example}
\newtheorem{lemma}{Lemma}
\newtheorem{property}{Property}
\newtheorem{remark}{Remark}
\newtheorem{theorem}{Theorem}
\newtheorem{Table}{Table}
\newtheorem{pf}{Proof}


\newcommand{\bbB}{{\mathbb{B}}}
\newcommand{\bbC}{{\mathbb{C}}}
\newcommand{\bbD}{{\mathbb{D}}}
\newcommand{\bbE}{{\mathbb{E}}}
\newcommand{\bbI}{{\mathbb{I}}}
\newcommand{\bbJ}{{\mathbb{J}}}
\newcommand{\bbK}{{\mathbb{K}}}
\newcommand{\bbN}{{\mathbb{N}}}
\newcommand{\bbO}{{\mathbb{O}}}
\newcommand{\bbR}{{\mathbb{R}}}
\newcommand{\bbS}{{\mathbb{S}}}
\newcommand{\bbT}{{\mathbb{T}}}
\newcommand{\bbU}{{\mathbb{U}}}
\newcommand{\bbV}{{\mathbb{V}}}
\newcommand{\bbX}{{\mathbb{X}}}
\newcommand{\bbY}{{\mathbb{Y}}}
\newcommand{\bbZ}{{\mathbb{Z}}}


\newcommand{\cA}{{\cal A}}         \newcommand{\dA}[1]{{\cal A}_{#1}}
\newcommand{\cB}{{\cal B}}         \newcommand{\dB}[1]{{\cal B}_{#1}}
\newcommand{\cE}{{\cal E}}
\newcommand{\cI}{{\cal I}}
\newcommand{\cJ}{{\cal J}}
\newcommand{\cK}{{\cal K}}
\newcommand{\cN}{{\cal N}}
\newcommand{\cO}{{\cal O}}
\newcommand{\cS}{{\cal S}}         \newcommand{\dS}[1]{{\cal S}_{#1}}
\newcommand{\cR}{{\cal R}}
\newcommand{\cT}{{\cal T}}
\newcommand{\cU}{{\cal U}}
\newcommand{\cV}{{\cal V}}
\newcommand{\cX}{{\cal X}}
\newcommand{\cY}{{\cal Y}}
\newcommand{\cZ}{{\cal Z}}


\newcommand{\f}[1]{f_{_{\small {\rm #1}}}}
\newcommand{\hf}[1]{\hat{f}_{#1}}
\newcommand{\tf}[2]{f_{{#1}}^{{#2}}}

\newcommand{\F}[1]{F_{_{\small {\rm #1}}}}

\newcommand{\h}[1]{h_{_{\tiny\rm #1}}}
\newcommand{\n}[1]{n_{_{\tiny\rm #1}}}
\newcommand{\del}[1]{\delta_{_{\tiny\rm #1}}}
\newcommand{\pii}[2]{\pi_{_{\tiny\rm #1\!, \, #2}}}
\newcommand{\q}[1]{q_{_{#1}}}

\newcommand{\tintB}{\dis\int\nolimits\hspace*{-0.45cm}\Box\,}
\newcommand{\intB}[1]{\dis\int_{#1}\hspace*{-0.43cm}\Box\;}

\newcommand{\yb}[1]{{\bar{y}^{\rm #1}}}
\newcommand{\Y}[1]{Y_{\rm #1}}
\newcommand{\YY}[1]{Y_{\! \mbox{\scriptsize #1}}}
\newcommand{\y}[2]{y_{#1}^{\rm #2}}

\newcommand{\no}[2]{\|#1\|_{_{#2}}}
\newcommand{\nor}[3]{\|#1\|_{_{#2}}^{#3}}

\newcommand{\bismall}[2]{(\!\!\begin{array}{c} #1 \\ \vspace*{-0.6cm} \\ #2 \end{array} \!\! )}
\newcommand{\bi}[2]{\Big(\!\!\begin{array}{c} #1 \\ #2 \end{array} \!\!\Big)}
\newcommand{\twoc}[2]{\bigg(\!\!\begin{array}{c} #1 \\ #2  \end{array} \!\!\bigg)}

\newcommand{\twolr}[2]{(\!\!\begin{array}{c} #1 \\ #2  \end{array} \!\!)}

\newcommand{\twob}[2]{\lbrace {#1 \atop #2} \rbrace}
\newcommand{\twoB}[2]{\bigg\{\!\!\begin{array}{c} #1 \\ #2  \end{array} \!\!\bigg\}}

\newcommand{\fourc}[4]{\bigg(\!\!\begin{array}{cc} #1 & #2 \\ #3 & #4 \end{array} \!\!\bigg)}
\newcommand{\fourlr}[4]{(\!\!\begin{array}{cc} #1 & #2 \\ #3 & #4 \end{array} \!\!)}

\newcommand{\threec}[3]{\left(\!\!\begin{array}{c} #1 \\ #2 \\ #3  \end{array}
                              \!\!\right)}
\newcommand{\threev}[3]{\left(\!\!\begin{array}{c} #1 \\   #2 \\   #3 \end{array}
                              \!\!\right)}

\newcommand{\threetwo}[6]{\left(\!\!\begin{array}{cc}
#1 & #2 \\
#3 & #4 \\
#5 & #6
\end{array} \!\!\right)}
\newcommand{\threethree}[9]{\left(\!\!\begin{array}{ccc}
#1 & #2 & #3 \\
#4 & #5 & #6 \\
#7 & #8 & #9
\end{array} \!\!\right)}

\newcommand{\Cases}[4]{\left\{\!\!\begin{array}{ll} #1, & #2,  \\ [3mm]
                                                    #3, & #4  \end{array} \right. }
\newcommand{\CASES}[5]{\left\{\!\!\begin{array}{ll} #1, & #2,  \\ #3
                                                    #4, & #5  \end{array} \right. }

\newcommand{\Casesthree}[6]{\left\{\!\!\begin{array}{ll} #1 & #2  \\ [2mm]
                                                         #3 & #4  \\ [2mm]
                                                         #5 & #6  \end{array} \right. }

\newcommand{\sr}[1]{\stackrel{#1}{=}}
\newcommand{\srge}[1]{\stackrel{#1}{\ge}}
\newcommand{\mc}[1]{\multicolumn{1}{c}{#1}}
\newcommand{\MC}[3]{\multicolumn{#1}{#2}{#3}}
\newcommand{\TC}[4]{\contentsline {#1}{\numberline {#2}#3}{#4}}
\newcommand{\LIST}[3]{\contentsline {section}{\numberline {\hspace*{-0.5cm}#1}#2}{#3}}

\newcommand{\Hpi}[1]{\hat{\pi}_{_{\tiny\rm #1}}}
\newcommand{\titem}[1]{\vspace*{-0.15cm}\item[{\rm #1} ]\hspace*{0.1cm}}
\newcommand{\ub}[1]{\underline{\bf #1}}
\newcommand{\RED}[1]{{\color{red} #1}}
\newcommand{\BLUE}[1]{{\color{blue} #1}}

\newcommand{\two}[2]{{#1}_{_{\tiny\rm #2}}}
\newcommand{\three}[3]{{#1}_{_{#2, \tiny\rm #3}}}

\newcommand{\ct}[1]{c^{(t)}_{#1}}


\newcommand{\0}{{\bf 0}\!\!\!\!{\bf 0}}
\newcommand{\1}{{\bf 1}\hspace*{-0.2cm}{\bf 1}}
\def\b1{{\bf 1\!\!\!1}}  


\newcommand{\ba}{{\bf a}}
\newcommand{\bb}{{\bf b}}
\newcommand{\bc}{{\bf c}}
\newcommand{\bd}{{\bf d}}
\newcommand{\bm}{{\bf m}}
\newcommand{\bn}{{\bf n}}
\newcommand{\bp}{{\bf p}}
\newcommand{\bq}{{\bf q}}
\newcommand{\br}{{\bf r}}
\newcommand{\bs}{{\bf s}}
\newcommand{\bt}{{\bf t}}
\newcommand{\bat}{{\bf At}}
\newcommand{\bu}{{\bf u}}
\newcommand{\bv}{{\bf v}}
\newcommand{\bw}{{\bf w}}
\newcommand{\bx}{{\bf x}}
\newcommand{\by}{{\bf y}}
\newcommand{\bz}{{\bf z}}

\newcommand{\bA}{{\bf A}}
\newcommand{\bB}{{\bf B}}
\newcommand{\bC}{{\bf C}}
\newcommand{\bD}{{\bf D}}
\newcommand{\bH}{{\bf H}}
\newcommand{\bI}{{\bf I}}
\newcommand{\bJ}{{\bf J}}
\newcommand{\bL}{{\bf L}}
\newcommand{\bM}{{\bf M}}
\newcommand{\bN}{{\bf N}}
\newcommand{\bO}{{\bf O}}
\newcommand{\bP}{{\bf P}}
\newcommand{\bQ}{{\bf Q}}
\newcommand{\bS}{{\bf S}}
\newcommand{\bT}{{\bf T}}
\newcommand{\bU}{{\bf U}}
\newcommand{\bV}{{\bf V}}
\newcommand{\bW}{{\bf W}}
\newcommand{\bX}{{\bf X}}
\newcommand{\bY}{{\bf Y}}
\newcommand{\bZ}{{\bf Z}}


\newcommand{\iba}{{\boldsymbol a}}   \newcommand{\ibA}{{\boldsymbol A}}
\newcommand{\ibb}{{\boldsymbol b}}   \newcommand{\ibB}{{\boldsymbol B}}
\newcommand{\ibc}{{\boldsymbol c}}   \newcommand{\ibC}{{\boldsymbol C}}
\newcommand{\ibe}{{\boldsymbol e}}
                                     \newcommand{\ibH}{{\boldsymbol H}}
\newcommand{\ibp}{{\boldsymbol p}}
\newcommand{\ibr}{{\boldsymbol r}}   \newcommand{\ibR}{{\bfm R}}
\newcommand{\ibs}{{\boldsymbol s}}
\newcommand{\ibt}{{\boldsymbol t}}

\newcommand{\ibu}{{\boldsymbol u}}   \newcommand{\ibU}{{\boldsymbol U}}
\newcommand{\ibv}{{\boldsymbol v}}   \newcommand{\ibV}{{\boldsymbol V}}
\newcommand{\ibw}{{\boldsymbol w}}   \newcommand{\ibW}{{\boldsymbol W}}
\newcommand{\ibx}{{\boldsymbol x}}   \newcommand{\ibX}{{\boldsymbol X}}
\newcommand{\iby}{{\boldsymbol y}}   \newcommand{\ibY}{{\boldsymbol Y}}
\newcommand{\ibz}{{\boldsymbol z}}


\newcommand{\hb}{\hat{b}}
\newcommand{\hh}{\hat{h}}
\newcommand{\hm}{\hat{m}}
\newcommand{\hp}{\hat{p}}
\newcommand{\hq}{\hat{q}}
\newcommand{\hR}{\hat{R}}
\newcommand{\hbr}{\hat{\br}}
\newcommand{\hth}{\hat{\theta}}          \newcommand{\hTh}{\hat{\Theta}}
\newcommand{\hbth}{{\boldsymbol{\hth}}}
\newcommand{\hbbe}{{\boldsymbol{\hbe}}}
\newcommand{\hbmu}{{\boldsymbol{\hmu}}}
\newcommand{\hbnu}{{\boldsymbol{\hnu}}}
\newcommand{\hbSi}{{\boldsymbol{\hSi}}}
\newcommand{\hbp}{{\boldsymbol{\hp}}}
\newcommand{\heta}{\hat{\eta}}
\newcommand{\hvth}{\hat{\vartheta}}
\newcommand{\hvphi}{\hat{\log}}
\newcommand{\hmu}{\hat{\mu}}
\newcommand{\hnu}{\hat{\nu}}
\newcommand{\hsi}{\hat{\sigma}}   \newcommand{\hSi}{\hat{\Sigma}}
\newcommand{\hal}{\hat{\alpha}}   \newcommand{\hbe}{\hat{\beta}}
\newcommand{\hga}{\hat{\gamma}}
\newcommand{\hpsi}{\hat{\psi}}
\newcommand{\hxi}{\hat{\xi}}
\newcommand{\hpi}{\hat{\pi}}
\newcommand{\hla}{\hat{\lambda}}  \newcommand{\hbla}{\hat{{\bfm \lambda}}}
\newcommand{\hrho}{\hat{\rho}}
\newcommand{\hphi}{\hat{\phi}}
\newcommand{\hbvth}{\hat{\bf \vartheta}}

\newcommand{\whCorr}{\widehat{\mbox{Corr}}}
\newcommand{\whVar}{\widehat{\mbox{Var}}}
\newcommand{\whse}{\widehat{\mbox{se}}}
\newcommand{\whPr}{\widehat{\Pr}}

\newcommand{\tD}{\tilde{D}}
\newcommand{\tp}{\tilde{p}}           \newcommand{\tbp}{\tilde{{\bfm p}}}
\newcommand{\tx}{\tilde{x}}
\newcommand{\ty}{\tilde{y}}
\newcommand{\tth}{\tilde{\theta}}     \newcommand{\tbth}{\tilde{{\bfm \theta}}}
\newcommand{\tpi}{\tilde{\pi}}
\newcommand{\tmu}{\tilde{\mu}}        \newcommand{\tbmu}{\tilde{\bmu}}
\newcommand{\tbe}{\tilde{\beta}}      \newcommand{\tbbe}{\tilde{{\bfm \beta}}}
\newcommand{\tsi}{\tilde{\sigma}}     \newcommand{\tSi}{\tilde{\Sigma}}
\newcommand{\tbSi}{\tilde{\bSi}}
\newcommand{\tpsi}{\tilde{\psi}}
\newcommand{\txi}{\tilde{\xi}}
\newcommand{\tla}{\tilde{\lambda}}
\newcommand{\tphi}{\tilde{\phi}}
\newcommand{\trho}{\tilde{\rho}}
\newcommand{\tnu}{\tilde{\nu}}

\newcommand{\wtF}{\widetilde{F}}

\newcommand{\Bu}{\bar{u}}
\newcommand{\Bv}{\bar{v}}
\newcommand{\Bw}{\bar{w}}
\newcommand{\Bx}{\bar{x}}
\newcommand{\By}{\bar{y}}       \newcommand{\BY}{\bar{Y}}
\newcommand{\Bz}{\bar{z}}
\newcommand{\Bxi}{\bar{\xi}}
\newcommand{\BVar}{\overline{\mbox{Var}}}


\newcommand{\al}{\alpha}
\newcommand{\be}{\beta}
\newcommand{\ga}{\gamma}            \newcommand{\Ga}{\Gamma}
\newcommand{\de}{\delta}            \newcommand{\De}{\Delta}
\newcommand{\la}{\lambda}           \newcommand{\La}{\Lambda}
\renewcommand{\th}{\theta}          

\newcommand{\thx}{\theta_x}         \newcommand{\Thx}{\Theta_x}
\newcommand{\thy}{\theta_y}
\newcommand{\si}{\sigma}            \newcommand{\Si}{\Sigma}

\newcommand{\ka}{\kappa}
\newcommand{\om}{\omega}            \newcommand{\Om}{\Omega}
\newcommand{\ve}{\varepsilon}
\newcommand{\vp}{\varphi}
\newcommand{\vr}{\varrho}
\newcommand{\vth}{\vartheta}

\newcommand{\tht}{\theta^{(t)}}
\newcommand{\mut}{\mu^{(t)}}
\newcommand{\sit}{\si^{(t)}}
\newcommand{\st}{s^{(t)}}


\newcommand{\bxi}{{\bfm \xi}}      \newcommand{\bet}{{\bfm \eta}}
\newcommand{\bphi}{{\bfm \phi}}    \newcommand{\bep}{{\bfm \epsilon}}
\newcommand{\bmu}{{\bfm \mu}}      \newcommand{\bmut}{\bmu^{(t)}}
\newcommand{\bnu}{\boldsymbol{\nu}}

\newcommand{\bla}{{\bfm \lambda}}  \newcommand{\bLa}{\boldsymbol{\Lambda}}
\newcommand{\bSi}{{\bfm \Sigma}}
\newcommand{\bom}{{\bfm \omega}}   \newcommand{\bOm}{{\bfm \Omega}}
\newcommand{\bde}{{\bfm \delta}}   \newcommand{\bDe}{{\bfm \Delta}}
\newcommand{\bth}{\boldsymbol{\theta}}    \newcommand{\bTh}{\boldsymbol{\Theta}}
\newcommand{\btht}{\bth^{(t)}}
\newcommand{\bsi}{{\boldsymbol \sigma}}
\newcommand{\bbe}{{\boldsymbol \beta}}
\newcommand{\bpsi}{{\bfm \psi}}          \newcommand{\bPsi}{\boldsymbol{\Psi}}
\newcommand{\bpi}{{\bfm \pi}}
\newcommand{\bve}{{\bfm \varepsilon}}
\newcommand{\bvth}{\bfm \vartheta}

                                        \newcommand{\bGa}{\boldsymbol{\Gamma}}


\newcommand{\ALaplace}{\mbox{ALaplace}}
\newcommand{\Bernoulli}{\mbox{Bernoulli}}
\newcommand{\Binomial}{\mbox{Binomial}}
\newcommand{\BBinomial}{\mbox{BBinomial}}
\newcommand{\BP}{\mbox{BP}}
\newcommand{\Categorical}{\mbox{Categorical}}
\newcommand{\D}{\mbox{D}}
\newcommand{\Degenerate}{\mbox{Degenerate}}
\newcommand{\DExponential}{\mbox{DExponential}}
\newcommand{\Dirichlet}{\mbox{Dirichlet}}
\newcommand{\DMultinomial}{\mbox{DMultinomial}}
\newcommand{\Exponential}{\mbox{Exponential}}
\newcommand{\Exp}{\mbox{Exp}}
\newcommand{\FDiscrete}{\mbox{FDiscrete}}
\newcommand{\FN}{\mbox{FN}}
\newcommand{\GD}{\mbox{GD}}
\newcommand{\GDirichlet}{\mbox{GDirichlet}}
\newcommand{\GLiouville}{\mbox{GLiouville}}
\newcommand{\GP}{\mbox{GP}}
\newcommand{\GPD}{\mbox{GPD}}
\newcommand{\GPZIP}{\mbox{GP\_ZIP}}
\newcommand{\GPoisson}{\mbox{GPoisson}}
\newcommand{\Hgeometric}{\mbox{Hgeometric}}
\newcommand{\HPP}{\mbox{HPP}}

\newcommand{\I}{\mbox{I}}
\newcommand{\IBeta}{\mbox{IBeta}}
\newcommand{\Ichi}{\mbox{I}\raisebox{0.5ex}{$\chi$}}
\newcommand{\IG}{\mbox{IG}}
\newcommand{\IGamma}{\mbox{IGamma}}
\newcommand{\IGaussian}{\mbox{IGaussian}}
\newcommand{\IWishart}{\mbox{IWishart}}
\newcommand{\Laplace}{\mbox{Laplace}}
\newcommand{\Liouville}{\mbox{Liouville}}
\newcommand{\Logistic}{\mbox{Logistic}}

\newcommand{\Lognormal}{\mbox{Lognormal}}
\newcommand{\mBeta}{\mbox{Beta}}
\newcommand{\mGamma}{\mbox{Gamma}}
\newcommand{\Multinomial}{\mbox{Multinomial}}
\newcommand{\MVT}{\mbox{MVT}}
\newcommand{\N}{\mbox{N}}
\newcommand{\NBinomial}{\mbox{NBinomial}}
\newcommand{\ND}{\mbox{ND}}
\newcommand{\NDirichlet}{\mbox{NDirichlet}}
\newcommand{\NHPP}{\mbox{NHPP}}
\newcommand{\NIW}{\mbox{NIW}}
\newcommand{\PIG}{\mbox{PIG}}
\newcommand{\Poisson}{\mbox{Poisson}}
\newcommand{\SN}{\mbox{SN}}
\newcommand{\TBeta}{\mbox{TBeta}}
\newcommand{\TN}{\mbox{TN}}
\newcommand{\Tt}{\mbox{Tt}}
\newcommand{\Wishart}{\mbox{Wishart}}
\newcommand{\ZIGP}{\mbox{ZIGP}}
\newcommand{\ZIP}{\mbox{ZIP}}
\newcommand{\ZOIP}{\mbox{ZOIP}}
\newcommand{\ZINB}{\mbox{ZINB}}
\newcommand{\StN}{\mbox{StN}}
\newcommand{\SLN}{\mbox{SLN}}
\newcommand{\GI}{\mbox{GI}}


\newcommand{\Corr}{\mbox{Corr}}
\newcommand{\Cov}{\mbox{Cov}}
\newcommand{\CV}{\mbox{CV}}
\newcommand{\data}{\mbox{data}}
\newcommand{\EM}{\mbox{EM}}
\newcommand{\KL}{\mbox{KL}}
\newcommand{\logit}{\mbox{logit}}
\newcommand{\median}{\mbox{median}}
\newcommand{\MSE}{\mbox{MSE}}
\newcommand{\rank}{\mbox{rank}\,}
\newcommand{\se}{\mbox{se}}
\newcommand{\Se}{\mbox{Se}}
\newcommand{\SE}{\mbox{SE}}
\newcommand{\tr}{\mbox{$\,$tr$\,$}}
\newcommand{\Var}{\mbox{Var}}

\newcommand{\col}{\mbox{: }}
\newcommand{\diag}{\mbox{diag}}
\newcommand{\qand}{\quad \mbox{and} \quad}
\newcommand{\qag}{\quad \mbox{against} \quad}
\newcommand{\qas}{\quad \mbox{as} \quad}
\newcommand{\qif}{\mbox{if }\; }
\newcommand{\qor}{\quad \mbox{or} \quad}
\newcommand{\qve}{\quad \mbox{versus} \quad}
\newcommand{\qwh}{\quad \mbox{where} \quad}

\newcommand{\RE}{\mbox{RE}}
\newcommand{\yes}{\mbox{yes}}
\newcommand{\No}{\mbox{no}}
\newcommand{\DPP}{\mbox{DPP}}
\newcommand{\rd}{\,\mbox{d}}
\newcommand{\e}{\mbox{e}}
\newcommand{\w}{\mbox{w}}
\renewcommand{\ge}{\geqslant}
\renewcommand{\le}{\leqslant}


\newcommand{\II}{I\hspace*{-0.07cm}I}
\newcommand{\III}{I\hspace*{-0.07cm}I\hspace*{-0.07cm}I}
\newcommand{\IR}{{I\!\! R}}
\newcommand{\IV}{I\!V}
\newcommand{\et}{{\it et al}.}
\newcommand{\Et}{{\it et al}.\,}
\newcommand{\trv}{r.v.\ \!}

\newcommand{\yikH}{y_{ik}^{\rm H}}
\newcommand{\ybH}{\bar{y}^{\rm H}}


\newcommand{\sd}{\stackrel{\rm d}{=}}
\newcommand{\sdt}{$ {\small $\sd$} $}
\newcommand{\heq}{\;\hat{=}\;}
\newcommand{\teq}{\stackrel{\triangle}{=}}
\newcommand{\iid}{\stackrel{{\rm iid}}{\sim}}
\newcommand{\tiid}{i.i.d.$\hspace*{0.08cm}$}
\newcommand{\ind}{\stackrel{{\rm ind}}{\sim}}
\newcommand{\dsim}{\stackrel{.}{\sim}}
\newcommand{\dis}{\displaystyle}
\newcommand{\tex}{\textstyle}
\newcommand{\cf}{cf.$\hspace*{0.1cm}$}
\newcommand{\T}{\!\top\!}
\newcommand{\na}{\nabla}
\newcommand{\noi}{\noindent}
\newcommand{\ra}{\rightarrow}
\newcommand{\pr}{\propto}
\newcommand{\eq}{\equiv}
\newcommand{\pa}{\partial}
\newcommand{\ol}{\overline}
\newcommand{\non}{\nonumber}
\newcommand{\ap}{\approx}
\newcommand{\Bot}{\;\bot\;}
\newcommand{\inde}{\upmodels}        
\newcommand{\btu}{\bigtriangleup}
\newcommand{\ddint}{\Box\!\!\!\!\!\!\!\int}


\newcommand{\vs}{\vspace*{-0.25cm}}
\newcommand{\vkl}{\vskip 0.10in}
\newcommand{\vkL}{\vskip 0.15in}
\newcommand{\vkU}{\vskip 0.30in}


\newcommand{\namelistlabel}[1]{\mbox{#1}\hfil}
\newenvironment{namelist}[1]{%
\begin{list}{}
       {\let \makelabel \namelistlabel
        \settowidth{\labelwidth}{#1}
        \setlength{\leftmargin}{1.1\labelwidth}   }
        }{%
\end{list} }

\def\bds{\begin{description} \itemsep=-\parsep \itemindent=-0.7 cm}
\def\eds{\end{description}}
\def\i{\item}


\begin{flushleft}
{\Large \bf The normalized expectation--maximization (N-EM) algorithm}
\vkU
\end{flushleft}

\begin{flushleft}
{\bf Guo-Liang TIAN$^{{\rm a}, \dag}$},   \ \
{\bf Xuanyu LIU$^{{\rm b}, \dag}$}   \ \  and \ \
{\bf Yuanfan ZHAO$^{{\rm a}, *}$}  \ \

\vkL
{\small \it $^{\rm a}$Department of Statistics and Data Science,
                      Southern University of Science and Technology, \\ Shenzhen 518055, Guangdong Province, P. R. China}
\vkL
{\small \it $^{\rm b}$Department of Statistics and Actuarial Science, School of Computing and Data Science, \\
                      The University of Hong Kong, Pokfulam Road, Hong Kong, P. R. China}

\vkL
{\small  $^{\dag}$Equally contributed} \\
{\small  $^*$Corresponding author's Email: \textsf{1243116@mail.sustech.edu.cn} }\\
{\small \textsf{22 July 2026 ZYF}   \quad  \textsf{22 July 2026 TGL}}
\end{flushleft}

%
%
%
%
%

\baselineskip 0.25in \vskip 0.05in \noi  {\bf Abstract}. Although the
\textit{expectation--maximization} (EM) algorithm is a powerful optimization tool in statistics, it can only be applied to missing/incomplete data problems or to problems with a latent-variable structure. It is well known that the introduction of latent variables (or the data augmentation) is an art; i.e., it could only be done case by case. In this paper, we propose a new algorithm, a so-called \textit{normalized EM} (N-EM) algorithm, for a class of log-likelihood functions with integrals. As an extension of the original EM algorithm, the N-EM algorithm inherits all advantages of EM-type algorithms and consists of three steps: normalization step (N-step), expectation step (E-step) and maximization step (M-step), where the N-step is to construct a \textit{normalized density function} (ndf), the E-step is to compute a well-established surrogate $Q$-function and the M-step is to maximize the $Q$-function as in the original EM algorithm. The ascent property, the best choice of the ndf, and those N-EM algorithms with a difficult M-step are also explored. By multiple real applications, we have shown that the N-EM algorithm can solve some problems which cannot be addressed by the EM algorithm. Next, for problems to which the EM can be applied (often case by case), the N-EM algorithm can be employed in a unified framework. Numerical experiments are performed and convergence properties are also established.

\vkL \noi {\bf Keywords}: EM algorithm; Jensen's inequality; MM algorithm; N-EM algorithm; Normalized density function; Optimal lower bound.

\baselineskip 0.30in
\setcounter{equation}{0}


\section{$\!\!\!\!\!\!\!\!$. Introduction}  

The \textit{expectation-maximization} (EM) algorithm (Dempster \Et 1977) is used to find (local) maximum likelihood parameters of a statistical model in cases where the equations cannot be solved directly. This is a general-purpose algorithm for maximum likelihood estimation in multitudinous situations best described as incomplete-data problems, where algorithms such as the Newton-Raphson method may turn out to be more complicated.

The EM algorithm can be put into use in problems including not only classic incomplete-data situations, where there are missing data (Fernandez 2010, Suesse \& Zammit-Mangion 2017), censored (Vaida \&  Liu 2009,  Dejardin \& Lesaffre 2013) or truncated observations (Tian \Et 2019), but also a wide variety of other cases, where the incompleteness of the data is not natural or obvious. These include statistical models such as random effects (Wang \Et 2007, Jose \& Burzykowski 2010), mixtures (Karlis \& Xekalaki 1999, Hung \& Chang-Chien 2017, Panic \Et 2020), log-linear models (Kim 1997) and latent variable structures (Xie \& Simpson 2015, Nakashima \Et 2020). The EM algorithm can help directly deal with a large amounts of maximum likelihood problems, or can simplify a complex problem in deriving its \textit{maximum likelihood estimates} (MLEs), this is because the complete-data likelihood usually has a relatively nice form for many statistical problems. What is more, due to its versatility and desirable properties, EM algorithm has been found applications in a broad range of statistical contexts and in almost all areas where statistical techniques are applied, and it has also been developed for PET and SPECT Data (Brailean \Et 1992, Kay 1997, Hoogerheide 2012), Gaussian process (Ranjan \Et 2016, Wu \& Ma 2018, Wu \& Ma 2019), cluster analysis (Chen \& Jiang 2008, Kadir \Et 2014), hidden markov models (Cappe 2011), quantile regression (Yang 2018), isoform quantification from RNA-seq data analysis (Deng \Et 2019), folded normal regression (Liu \Et 2020) and survival data (McLachlan \& Krishnan 2007).

Although the EM algorithm has become a standard tool in statistician's repertoire, it is not free of imperfection, many of which came to light in attempting to apply it in certain complex incomplete-data problems and some even in innocuously simple incomplete-data problems. Accordingly, several modifications and extensions of the algorithm had been developed to partially overcome of these limitations. The developments were in the orientation of iterative simulation techniques, such as the Monte Carlo E-step (Wei \& Tanner 1990), Stochastic EM algorithm (Celeux 1985, Celeux \& Diebolt 1986), Data Augmentation ( Rubin 1987, Wei \& Tanner 1990a), and the Gibbs sampler (Casella \Et George 1992). It is worth noting that the EM algorithm has many similarities and connections with the Markov chain Monte Carlo algorithm (Ripley 2009), especially EM with data augmentation and Gibbs sampling.

However, there still exist some obstacles for users to formulate the incompleteness in an appropriate manner to facilitate the application of the EM algorithm. One of such obstacles is the inability of ingenious expressions of a given problem as an unobvious incomplete-data structure within EM framework. What is worse, some found that inappropriate missing data structure could even bog down the problems. In other words, the EM algorithm can only be applied to such likelihoods in the form of
$$
    L(\bth | \Y{obs}) = f(\Y{obs} | \bth) = \int_{\bbZ} f(\Y{obs}, \ibz | \bth) \rd \ibz,
$$
where $f(\Y{obs},\ibz | \bth)$ must be the joint density of the complete data, and $\ibz$ is the missing data.

Therefore, finding unobvious missing-data structures is the most urgent challenge and just the drawback of the original EM algorithm, which even needs some afflatus. Even if we could find a valid missing structure, we need design it case by case, lacking a unified approach to solve this problem systematically.
In addition, the existing algorithms can not solve the objective function with integral except numerical algorithm such as Newton-Cotes and similar numerical quadrature methods, or Monte Carlo integration.

The main objective of this paper is to propose a novel and instructive \textit{ normalized expectation-maximization} (N-EM) algorithm for estimating parameters of a complex likelihood function with an integral term or finite homogeneous summation, which provides a unified framework to solve a class of problems instead of designing case by case. More importantly, this is a formula-derived method to establish the EM algorithm, avoiding the difficulty of constrcuting missing structures skillfully, which has laid a solid foundation for further promotion and wider use of EM algorithms. On the other hand, this algorithm brings a new perspective to optimize the objective function with integral systematically. Compared with the ordinary methods, this algorithm has its superiority in relieving computational burden.

The N-EM algorithm consists of the normalization step (N-step), expectation step (E-step) and maximization step (M-step). E-step is to modify the surrogate function $Q$-function to maximize the lower bound, under the condition that the parameters are fixed. And, N-step is to calculate the pivotal functional parameter in the $Q$-function of E-step. We use M-step to optimize the $Q$-function with respect to parameter $\bth$. More importantly, the optimality of the normalized density function found by N-step and the ascent property show the theoretical properties of algorithm. Furthermore, we can use the N-EM algorithm theory \RED{to interpret} the EM algorithm, especially, the conditional predictive density.

In the section of applications, we divide into three subsections, respectively, solving some problems for the first time by applying the N-EM algorithm; unifying three kinds of traditional EM algorithms in the same N-EM framework and the case of finite discrete summation. We use these applications to fully illustrate that with obvious structural features, we can quickly judge the applicability of N-EM algorithms. Then, we can simply propose N-EM algorithms instead of searching for nebulous missing-structures. Furthermore, the proposed N-EM algorithm can solve a large number of problems that can not be solved by traditional EM algorithms. Those adequately illustrate the many specialties of N-EM algorithms such as easy to judge, formula-driven, wide-applied, and simple-designed, which is a unified framework of EM algorithms for different kinds of drivers.

Finally, some numerical experiments are conducted in order to evaluate the performance of the N-EM algorithm for different distributions and models.

The rest of the paper is organized as follows. The N-EM algorithm is introduced in Section 2, including the definition of the finite discrete integral, methodology of  integral, the finite discrete integral and the integral term in the denominator, and the N-EM Interpretation of EM. Section 3 introduces the properties, including the ascent property of algorithm, and the optimality of the normalized density function. Section 4 introduces the applications in three subsections: solving some problems for the first time by applying the N-EM algorithm; unifying three kinds of ordinary EM algorithms in the same N-EM framework and the case of finite discrete summation. In Section 5, we investigate the theoretical behaviors of the N-EM algorithms such as the local and global convergences. Numerical experiments are conducted in Section 6 to assess their practical performance.
\section{$\!\!\!\!\!\!\!\!$. The normalized EM algorithm and \\ \hspace*{-0.35cm} its monotonic convergence}  

Let an $m$-dimensional random vector $\bx = (X_1, \ldots, X_m)^{\T}$ have the joint \textit{probability density function} (pdf) $f(\ibx; \bth)$, where $\ibx = (x_1, \ldots, x_m)^{\T} \in \bbR^m$ is the realization of $\bx$, $\bth = (\th_1, \ldots, \th_p)^{\T}$ is the parameter vector belonging to $\bTh$, and $\bTh \subseteq \bbR^p$ is the parameter space. Assume that the joint pdf $f(\cdot)$ is of the following form
\begin{eqnarray} \label{eqn2.A1}
   f(\ibx; \bth) = c(\ibx, \bth) \times \int_{\bbS} h(\ibs | \ibx, \bth) \rd \ibs,
\end{eqnarray}
where $c(\ibx, \bth)$ is a function of both $\ibx$ and $\bth$; $h(\cdot | \ibx, \bth) \ge 0$, $\int_{\bbS} h(\ibs | \ibx, \bth) \rd \ibs < \infty$, and $h(\cdot| \ibx, \bth)$ is not necessarily being a joint density. In this paper, we consider three cases.
\begin{namelist}{0123456}
\item[\hspace*{-0.08cm} Case I:]  $c(\ibx, \bth)$ is relatively simple and the integral term $\int_{\bbS} h(\ibs | \ibx, \bth) \rd \ibs$ appears in the numerator of $f(\ibx; \bth)$.

\item[\hspace*{-0.08cm} Case \II:] $c(\ibx, \bth)$ is very complicated.

\item[\hspace*{-0.08cm} Case \III:] The integral term $\int_{\bbS} h(\ibs | \ibx, \bth) \rd \ibs$ appears in the denominator of $f(\ibx; \bth)$.
\end{namelist}

\subsection{Formulation of the N-EM algorithm for Case I}  

Let random vectors $\bx_1, \ldots, \bx_n \iid f(\ibx; \bth)$ and $\Y{obs} = \{\ibx_i \}_{i=1}^n$ denote the observed data, where $\ibx_i = (x_{i1}, \ldots, x_{im})^{\T}$ is the realization of $\bx_i = (X_{i1}, \ldots, X_{im})^{\T}$, then the log-likelihood function of $\bth$ can be written as
\begin{eqnarray} \label{eqn2.1}
   \ell(\bth | \Y{obs}) = \sum_{i=1}^n \log [c(\ibx_i, \bth)] + \sum_{i=1}^n \log \left[\int_{\bbS} h(\ibs | \ibx_i, \bth ) \rd \ibs \right].
\end{eqnarray}

To calculate the MLEs $\hbth$ of $\bth$, similar to the original EM algorithm with two steps (i.e., E- and M-steps), we propose an N-EM algorithm with following three steps (i.e., normalizing the function $h(\cdot| \ibx, \bth)$ in the N-step, calculating the $Q$-function in the E-step, and maximizing the $Q$-function with respect to $\bth$ in the M-step):
\begin{namelist}{01234567}
\item[\hspace*{-0.08cm} N-step:] By normalizing the function $h(\ibs | \ibx_i, \bth)$, we first construct a so-called \textit{normalizing density function} (ndf) $g(\ibs | \ibx_i, \bth) = h(\ibs | \ibx_i, \bth) \big/\int_{\bbS} h(\ibs' | \ibx_i, \bth) \rd \ibs'$ for $\ibs \in \bbS$, so that
    \begin{eqnarray} \label{eqn2.2}
       g(\ibs | \ibx_i, \btht) = \frac{h(\ibs | \ibx_i, \btht)}{\dis \int_{\bbS} h(\ibs' | \ibx_i, \btht) \rd \ibs'}, \qquad \ibs \in \bbS,
    \end{eqnarray}
    is the ndf of $h(\ibs | \ibx_i, \btht)$, where $\btht$ denotes the $t$-th approximation of $\hbth$.

\item[\hspace*{-0.08cm} E-step:] To calculate the $Q$-function as follows:
    \begin{eqnarray} \label{eqn2.3}
       Q(\bth | \btht) \teq \sum_{i=1}^n \log [c(\ibx_i, \bth)] + \sum_{i=1}^n \int_{\bbS} \log [h(\ibs | \ibx_i, \bth)] \cdot g(\ibs | \ibx_i, \btht) \rd  \ibs,
    \end{eqnarray}
    where the $i$-th integral is exactly the \textit{expectation} of $\log [h(\bs | \ibx_i, \bth)]$ with respect to the density $g(\ibs | \ibx_i, \btht)$.

\item[\hspace*{-0.08cm} M-step:] Given $\btht$, by maximizing the $Q$-function with respect to $\bth$, we update $\btht$ by
    \begin{eqnarray} \label{eqn2.TGL5}
       \bth^{(t+1)}=\arg \, \max _{\bth \in \bTh} Q(\bth | \btht).
    \end{eqnarray}
\end{namelist}
The N-, E- and M-steps are alternately repeated till $|\ell(\bth^{(t+1)} | \Y{obs}) - \ell(\btht | \Y{obs})| \le \de_0$, where $\de_0$ is a pre-specified small positive number.

In the following subsections, we shall address four important issues: (i) How to choose the ndf $g(\cdot | \ibx_i, \bth)$ and why is (\ref{eqn2.2}) the best choice? see \S 2.3;  (ii) How to derive the $Q$-function specified in (\ref{eqn2.3}) and how to prove the monotone convergence of the N-EM algorithm when $\bth^{(t+1)}$ is defined by (\ref{eqn2.TGL5})? see \S 2.2; (iii) Why is the original EM algorithm a special case of the N-EM algorithm? see \S 2.4; (iv) How to deal with difficult M-steps in the N-EM algorithm for Cases \II--\III? see \S 2.5.

\subsection{Ascent property of the N-EM algorithm}  

To derive the ascent property of the N-EM algorithm (\ref{eqn2.2})--(\ref{eqn2.TGL5}), we first briefly review the MM principle, the integral version of Jensen's inequality and its one drawback.

\subsubsection{MM principle and optimal lower bound} 

Let $\ell(\bth | \Y{obs})$ be a concave function and we want to find its maximizer
\begin{equation} \label{eqn2.TGL6}
   \hbth = \arg \, \max_{\bth \in \bTh} \ell(\bth | \Y{obs}).
\end{equation}
The \textit{minorization--maximization} (MM) technique is a general principle of constructing monotonic algorithms (Becker \Et 1997, Lange \Et 2000, Hunter \& Lange 2004, Tian \Et 2019) for the optimization problem (\ref{eqn2.TGL6}). The basic idea is as follows. First, one needs to find a minorizing function $Q^*(\bth | \btht)$ satisfying two conditions:
\begin{equation} \label{eqn3.1}
   Q^*(\bth | \btht) \le \ell(\bth | \Y{obs}), \quad \forall \bth, \btht \in \bTh \qand Q^*(\btht | \btht)=\ell(\btht | \Y{obs}).
\end{equation}
In other words, for any given $\btht$, the $Q^*(\cdot| \btht)$ function always lies under $\ell(\cdot |\Y{obs})$ and is tangent to $\ell(\cdot |\Y{obs})$ at the point $\bth=\btht$. The $Q^*(\cdot| \btht)$ function satisfying (\ref{eqn3.1}) is also called the \textit{optimal lower bound} (OLB) of the objective function $\ell(\bth |\Y{obs})$ with the tangent point being $\bth=\btht$. Second, instead of maximizing the objective function $\ell(\bth | \Y{obs})$, one maximizes the OLB $Q^*(\cdot | \btht)$ to obtain its maximizer
\begin{equation} \label{eqn3.2}
   \bth^{(t+1)} = \arg \, \max _{\bth \in \bTh} Q^*(\bth | \btht).
\end{equation}
Subsequently, we have
\begin{equation} \label{eqn3.3}
   \ell(\bth^{(t+1)} | \Y{obs}) \ge Q^*(\bth^{(t+1)} | \btht) \ge  Q^*(\btht | \btht)=\ell(\btht | \Y{obs}).
\end{equation}
The MM iterate (\ref{eqn3.2}) possesses the ascent property deriving the objective  function $\ell(\bth  | \Y{obs})$ uphill. Under appropriate conditions of compactness and continuity, the ascent property (\ref{eqn3.3}) guarantees the convergence of the MM technique, and lends the algorithm monotone convergence.

\subsubsection{Integral version of Jensen's inequality and its one drawback} 

Let $X$ be a random variable taking values in the domain of a concave function $h(\cdot)$, Jensen's inequality states that $h[E(X)] \ge E[h(X)]$, provided that both $E(X)$ and $E[g(X)]$ exist. The continuous/integral version of Jensen's inequality is
\begin{equation} \label{eqn2.NEM1}
   h\left[\int_{\bbX} \tau(x) \cdot g(x) \rd x \right] \ge
   \int_{\bbX} h[\tau(x)] \cdot g(x) \rd x,
\end{equation}
where $\bbX \subseteq \bbR$, $\tau(\cdot)$ is an arbitrary real function, and $g(\cdot)$ is a pdf defined in $\bbX$. However, the inequality (\ref{eqn2.NEM1}) does not indicate that under what kind of conditions, the equality holds. The usefulness of (\ref{eqn2.NEM1}) can be found in the next subsection.

\subsubsection{Monotone convergence of the N-EM algorithm} 

This subsection shows that the N-EM algorithm (\ref{eqn2.2})--(\ref{eqn2.TGL5}) satisfies the two conditions in (\ref{eqn3.1}), i.e., the N-EM algorithm possesses the ascent property deriving the log-likelihood function $\ell(\bth| \Y{obs})$ uphill. By using the integral version of Jensen's inequality (\ref{eqn2.NEM1}), we can prove the first condition in  (\ref{eqn3.1}) as follows:
\begin{eqnarray}
   \ell(\bth |\Y{obs}) &\sr{(\ref{eqn2.1})} & \sum_{i=1}^n \log [c(\ibx_i, \bth)]+\sum_{i=1}^n \log \left[\int_{\bbS} h(\ibs | \ibx_i,\bth) \rd \ibs\right] \non \\ [2mm]
   &=& \sum_{i=1}^n \log [c(\ibx_i, \bth)] + \sum_{i=1}^n \log \left[\int_{\bbS} \frac{h(\ibs | \ibx_i, \bth)} {g(\ibs | \ibx_i, \btht)} \cdot g(\ibs | \ibx_i, \btht)\rd \ibs\right]  \qquad \label{eqn2.NEMtgl10} \\ [2mm]
   &\srge{(\ref{eqn2.NEM1})} & \sum_{i=1}^n \log [c(\ibx_i, \bth)] + \sum_{i=1}^n \int_{\bbS} \log \left[ \frac{h(\ibs | \ibx_i, \bth)} {g(\ibs | \ibx_i, \btht)}\right] \cdot g(\ibs | \ibx_i, \btht) \rd \ibs \non \\ [2mm]
   &\teq & Q^*(\bth|\btht)  \label{eqn2.NEM10}   \\ [2mm]
   &\sr{(\ref{eqn2.3})} & Q(\bth|\btht)-\sum_{i=1}^n \int_{\bbS} \log \left[g(\ibs | \ibx_i, \btht)\right] \cdot g(\ibs | \ibx_i, \btht) \rd \ibs, \non
\end{eqnarray}
where the last term on the right-hand side does not depend on $\bth$ and only depends on $\{\ibx_i\}_{i=1}^n$ and $\btht$, so that the difference between $Q^*(\bth|\btht)$ and $Q(\bth|\btht)$ in E-step of the N-EM algorithm is a constant not depending on $\bth$. Thus, we have
$$
   \bth^{(t+1)}=\arg \, \max _{\bth \in \bTh} Q^*(\bth | \btht) = \arg \, \max_{\bth \in \bTh} Q(\bth | \btht),
$$
which is the M-step (\ref{eqn2.TGL5}) exactly.  In (\ref{eqn2.NEM10}), let $\bth=\btht$, we obtain
\begin{eqnarray} \label{eqn2.NEM12}
   &   & Q^*(\btht|\btht) \non \\ [2mm]
   & \sr{(\ref{eqn2.NEM10})} & \sum_{i=1}^n \log [c(\ibx_i, \btht)] + \sum_{i=1}^n \int_{\bbS} \log \left[ \frac{h(\ibs | \ibx_i, \btht)}{g(\ibs | \ibx_i, \btht)}\right] \cdot g(\ibs | \ibx_i, \btht) \rd \ibs \non \\ [2mm]
   & \sr{(\ref{eqn2.2})} & \sum_{i=1}^n \log [c(\ibx_i, \btht)] + \sum_{i=1}^n \int_{\bbS} \log \left[ \frac{h(\ibs | \ibx_i, \btht)} {h(\ibs | \ibx_i, \btht)} \int_{\bbS} h(\ibs' | \ibx_i, \btht) \rd \ibs' \right] \cdot g(\ibs | \ibx_i, \btht) \rd \ibs   \quad \non \\ [2mm]
   &=& \sum_{i=1}^n \log [c(\ibx_i, \btht)] + \sum_{i=1}^n \log \left[ \int_{\bbS} h(\ibs | \ibx_i, \btht) \rd \ibs \right]  \non \\ [2mm]
   &\sr{(\ref{eqn2.1})} & \ell(\btht | \Y{obs}).
\end{eqnarray}
implying that the second condition in (\ref{eqn3.1}) is satisfied. Hence, the ascent property (\ref{eqn3.3}) holds for $\ell(\cdot | \Y{obs})$ and $Q^*(\cdot |\btht)$. Thus, under appropriate conditions of compactness and continuity, the ascent property guarantees the convergence of the proposed N-EM algorithm (\ref{eqn2.2})--(\ref{eqn2.TGL5}).

\subsection{Best choice of the $g(\ibs | \ibx_i, \btht)$ function}  

From (\ref{eqn2.NEM10})--(\ref{eqn2.NEM12}), we can see that $Q^*(\bth | \btht)$ defined by (\ref{eqn2.NEM10}) is the OLB of the log-likelihood function $\ell(\bth |\Y{obs})$ with the tangent point being $\bth=\btht$. The most important idea for the N-EM algorithm is to choose the best ndf $g(\ibs | \ibx_i, \btht)$ specified by (\ref{eqn2.2}) such that $Q^*(\bth |\btht)$ is the OLB of $\ell(\bth |\Y{obs})$, resulting in a faster convergence speed. This subsection will prove that the ndf $g(\ibs | \ibx_i, \btht)$ in (\ref{eqn2.2}) is the best choice to provide the OLB $Q^*(\bth|\btht)$.

In fact, if we replace $g(\ibs | \ibx_i, \btht)$ in (\ref{eqn2.NEMtgl10}) by any pdf $q_{\ibx_i}(\ibs)$ on the same support $\bbS$ for all $i=1, \ldots, n$, then (\ref{eqn2.NEMtgl10}) is still true; that is
\begin{eqnarray*}
   \ell(\bth | \Y{obs}) &=& \sum_{i=1}^n \log [c(\ibx_i, \bth)] + \sum_{i=1}^n \log \left[\int_{\bbS} \frac{h(\ibs | \ibx_i, \bth)}{q_{\ibx_i}(\ibs)} \cdot q_{\ibx_i}(\ibs) \rd \ibs \right]  \\ [2mm]
   &\srge{(\ref{eqn2.NEM1})} & \sum_{i=1}^n \log [c(\ibx_i, \bth)] + \sum_{i=1}^n \int_{\bbS} \log \left[ \frac{h(\ibs | \ibx_i, \bth)} {q_{\ibx_i}(\ibs)} \right] \cdot q_{\ibx_i}(\ibs) \rd \ibs  \\ [2mm]
   & \teq &  F (\Y{obs}, \bth | q_{\ibx_1}(\ibs), \ldots, q_{\ibx_n}(\ibs) ),
\end{eqnarray*}
which is a class of lower bounds of $\ell(\bth | \Y{obs})$.

We want to find the best expressions of $\{q_{\ibx_i}(\ibs)\}_{i=1}^n$, denoted by $\{q_{\ibx_i}^{(t)}(\ibs)\}_{i=1}^n$   depending on $\btht$, such that $F(\Y{obs}, \bth | q_{\ibx_1}^{(t)}(\ibs), \ldots, q_{\ibx_n}^{(t)}(\ibs))$ is the OLB of $\ell(\bth | \Y{obs})$ with the tangent point being $\bth=\btht$. In other words, we need to find the best pdfs
\begin{eqnarray} \label{eqn2.NEM13}
   \left\{ q^{(t)}_{\ibx_i}(\ibs)\right\}_{i=1}^n &=& \arg  \max_{\{q_{\ibx_i}(\ibs)\}_{i=1}^n} F \left(\Y{obs}, \btht | q_{\ibx_1}(\ibs), \ldots, q_{\ibx_n}(\ibs) \right) \qquad \\ [2mm]
   & & \mbox{subject to } \int_{\bbS} q_{\ibx_i}(\ibs) \rd \ibs =1, \quad i=1, \ldots, n.  \non
\end{eqnarray}
By the Lagrange multiplier method, the constrained optimization problem (\ref{eqn2.NEM13}) can be changed into the unconstrained optimization problem:
\begin{equation} \label{eqn2.NEM14}
   \left\{ q^{(t)}_{\ibx_i}(\ibs), \la_i^{(t)} \right\}_{i=1}^n = \arg  \max_{\{q_{\ibx_i}(\ibs),\; \la_i\}_{i=1}^n} \wtF \left(\Y{obs}, \btht | q_{\ibx_1}(\ibs), \ldots, q_{\ibx_n}(\ibs); \la_1, \ldots, \la_n \right),
\end{equation}
where $\{\la _i\}_{i=1}^n$ are Lagrange multipliers with $\la_i >0$, and
\begin{eqnarray*}
   & & \wtF \left(\Y{obs}, \btht | q_{\ibx_1}(\ibs), \ldots, q_{\ibx_n}(\ibs); \la_1, \ldots, \la_n \right) \\ [2mm]
   &\teq & F \left(\Y{obs}, \btht | q_{\ibx_1}(\ibs), \ldots, q_{\ibx_n}(\ibs) \right)  - \sum_{i=1}^n \la _i\left[ \int_{\bbS} q_{\ibx_i}(\ibs) \rd \ibs-1\right] \\ [2mm]
   &=& \sum_{i=1}^n \left\{ \log [c(\ibx_i, \btht)] + \int_{\bbS} \log \left[ \frac{h(\ibs | \ibx_i, \btht)} {q_{\ibx_i}(\ibs)} \right] \cdot q_{\ibx_i}(\ibs) \rd \ibs - \la _i\left[ \int_{\bbS} q_{\ibx_i}(\ibs) \rd \ibs-1\right] \right\} \\ [2mm]
   &\teq & \sum_{i=1}^n \wtF_i(\ibx_i, \btht| q_{\ibx_i}(\ibs), \la_i)
\end{eqnarray*}
is an additive separable function with respect to $\{\ibx_i, q_{\ibx_i}(\ibs), \la_i\}_{i=1}^n$. Thus, the complicated optimization problem (\ref{eqn2.NEM14}) is equivalent to $n$ simpler optimization problems:
$$
   \left\{ q^{(t)}_{\ibx_i}(\ibs), \la_i^{(t)} \right\} = \arg  \max_{\{q_{\ibx_i}(\ibs),\; \la_i\} } \wtF_i(\ibx_i, \btht| q_{\ibx_i}(\ibs), \la_i), \quad i=1, \ldots, n,
$$
which can be maximized separately. By functionally partially differentiating $\wtF_i$ with respect to $q_{\ibx_i}(\ibs)$ and $\la_i$ with the Euler--Lagrange equation, we obtain
$$
   q_{\ibx_i}(\ibs) = \e^{-1 - \la_i} \cdot h(\ibs | \ibx_i,\btht) \qand \la _i=-1 + \log \left[\int_{\bbS} h(\ibs | \ibx_i,\btht) \rd \ibs \right],
$$
so that for $i=1, \ldots, n$,
$$
   q_{\ibx_i}^{(t)}(\ibs) = \frac{h(\ibs | \ibx_i, \btht)}{\int_{\bbS} h(\ibs' | \ibx_i, \btht) \rd \ibs'} \sr{(\ref{eqn2.2})} g(\ibs | \ibx_i, \btht),  \qquad \ibs \in \bbS.
$$
In other words, the pdfs $\{q_{\ibx_i}^{(t)}(\ibs)\}_{i=1}^n$ can produce the OLB $F(\Y{obs}, \bth | q_{\ibx_1}^{(t)}(\ibs), \ldots, q_{\ibx_n}^{(t)}(\ibs))$ at $t$-th iteration of the N-EM algorithm.

\subsection{The original EM algorithm is a special case of \\ the N-EM algorithm}  

In the original EM algorithm (Dempster \Et 1977) with a missing-data structure, the E-step is to compute the surrogate function $Q^{\rm EM}(\bth | \btht)$ defined by
\begin{equation} \label{eqn2.NEM15}
   Q^{\rm EM}(\bth | \btht) = \int_{\bbZ} \ell(\bth | \Y{obs}, \ibz) \cdot f(\ibz | \Y{obs}, \btht) \rd  \ibz,
\end{equation}
where $\Y{obs}$ denotes the observed data, $\ibz$ is the latent vector (or missing data), $\ell(\bth | \Y{obs}, \ibz)$ is the log-likelihood function for the complete data and $f(\ibz | \Y{obs}, \bth)$ is the conditional predictive density. The M-step is to maximize $Q^{\rm EM}(\bth | \btht)$ about $\bth$ to obtain its maximizer
$$
   \bth^{(t+1)} = \arg\, \max_{\bth \in \bTh} Q^{\rm EM}(\bth | \btht).
$$
The two-step process is repeated until convergence occurs. We wonder how the three Harvard Professors originally derived the $Q^{\rm EM}(\bth | \btht)$ function defined by (\ref{eqn2.NEM15}). Our guess is that they utilized the integral version of Jensen's inequality which plays a crucial role in the process of EM derivation. This subsection will show that the EM algorithm is a special case of the proposed N-EM algorithm, confirming our guess.

In fact, the joint density of the observed data $\Y{obs} = \{\ibx_i\}_{i=1}^n$ is
$$
   f(\Y{obs} | \bth)=\int_{\bbZ} f(\Y{obs}, \ibz | \bth) \rd  \ibz=L(\bth| \Y{obs}),
$$
so that the observed-data log-likelihood function of $\bth$ is given by
$$
   \ell(\bth | \Y{obs})=\log [L(\bth | \Y{obs})] = \log \left[\int_{\bbZ} f(\Y{obs}, \ibz | \bth) \rd  \ibz\right],
$$
which has the form of (\ref{eqn2.1}). From the N-step of the N-EM algorithm, we take $g(\ibz | \Y{obs}, \bth)$ as the normalized density of the joint density $f(\Y{obs}, \ibz | \bth)$ of the complete data, i.e.,
\begin{equation} \label{eqn2.NEM16}
   g(\ibz | \Y{obs}, \bth) = \frac{f(\Y{obs}, \ibz | \bth)} {\dis \int_{\bbZ} f(\Y{obs}, \ibz' | \bth) \rd  \ibz'} =\frac{f(\Y{obs}, \ibz | \bth)} {f(\Y{obs} | \bth)} = f(\ibz | \Y{obs}, \bth),
\end{equation}
which is exactly the conditional predictive density. From the E-step of the N-EM algorithm, the $Q$-function can be established by
\begin{eqnarray*}
   Q(\bth | \btht) & \sr{(\ref{eqn2.3})} &  \int_{\bbZ} \log \left[f(\Y{obs}, \ibz | \bth) \right] \cdot g(\ibz | \ibx_i, \btht) \rd  \ibz \\ [2mm]
   &\sr{(\ref{eqn2.NEM16})} & \int_{\bbZ} \ell(\bth | \Y{obs}, \ibz) \cdot f(\ibz | \Y{obs}, \btht) \rd  \ibz,
\end{eqnarray*}
which is exactly the $Q^{\rm EM}(\bth | \btht)$ function defined in the original EM algorithm.

Finally, as shown in (\ref{eqn2.NEM16}), the fact of the conditional predictive density being the ndf of the joint density of the complete data, can be regraded as the essence of the original EM being a special case of the proposed N-EM algorithm.

\subsection{The N-EM algorithm with a difficult M-step for Cases \II--\III}  

For Case \II \ and Case \III, directly maximizing the $Q$-function (\ref{eqn2.3}) will result in that the M-step (\ref{eqn2.TGL5}) of the N-EM algorithm does not have explicit solutions. For such cases, we can utilize some useful inequalities (especially, Jensen's inequalities) to create a new minorizing function $Q^{\rm New}(\bth | \btht)$ as the OLB of the $Q(\bth | \btht)$ function with tangent point being $\bth=\btht$.

\subsubsection{Supporting hyperplane inequality} 

Let $\vp(z)$ is a \textit{convex} function defined on a convex set $\bbC$, i.e., $\vp''(z)\ge 0$ for all $z \in \bbC$. The supporting hyperplane inequality states that
\begin{equation} \label{eqn2.NEM17}
   \vp(z) \ge \vp(z_0) + (z-z_0)\vp'(z_0), \quad \forall z, z_0 \in \bbC,
\end{equation}
where the equality holds \textit{if and only if} (iff) $z=z_0$. The usefulness of (\ref{eqn2.NEM17}) can be found in Subsections 3.2.2--3.2.3.

\subsubsection{Discrete version of Jensen's inequality and beyond} 

Let $\phi(z)$ is a \textit{concave} function defined on a convex set $\bbC$, i.e., $\phi''(z)\le 0$ for all $z \in \bbC$. The discrete version of Jensen's inequality is
\begin{equation} \label{eqn2.NEM18}
   \phi \left( \sum_{k=1}^K \al_k z_k \right) \ge \sum_{k=1}^K \al_k \phi(z_k),
\end{equation}
which is true for any probability weights $\{\al_k\}_{k=1}^K$ satisfying: $\al_k >0$ and $\sum_{k=1}^K \al_k=1$. However, the inequality (\ref{eqn2.NEM18}) does not show that under what kind of conditions, the equality holds.

In the follows, for $K=2$, we can provide a sufficient and necessary condition such that the equality in (\ref{eqn2.NEM18}) holds. In fact, for $K=2$, we have
\begin{eqnarray*}
   \phi\big(u(\bth) + v(\bth) \big) &\ge &  \frac{u(\btht)}{u(\btht) + v(\btht)} \phi \left( \frac{u(\btht) + v(\btht)}{u(\btht)} u(\bth) \right) \\ [2mm]
   & &  +\; \frac{v(\btht)}{u(\btht) + v(\btht)} \phi \left( \frac{u(\btht) + v(\btht)}{v(\btht)} v(\bth) \right),
\end{eqnarray*}
where the equality holds iff $\bth=\btht$.

For the special case $\phi(z) = \log(z)$, if we assume $u(\bth) >0$ and $v(\bth) >0$, then we have
\begin{equation} \label{eqn2.NEM19}
   \log\left[u(\bth) + v(\bth) \right] \ge  \frac{u(\btht)}{u(\btht) + v(\btht)} \log \left[ u(\bth) \right] + \frac{v(\btht)}{u(\btht) + v(\btht)} \log \left[ v(\bth) \right] + \ct{0},
\end{equation}
or
\begin{equation} \label{eqn2.NEM20}
    - \log \left[ u(\bth) \right] \ge \frac{v(\btht)}{u(\btht)} \log \left[ v(\bth) \right] - \frac{u(\btht) + v(\btht)}{u(\btht)} \log\left[u(\bth) + v(\bth) \right] + \ct{01},
\end{equation}
where $\ct{0}$ and $\ct{01}$ are two constants free from $\bth$. Especially, when $u(\bth) + v(\bth) = C$ (i.e., a constant not depending on $\bth$), then (\ref{eqn2.NEM20}) becomes
\begin{equation} \label{eqn2.NEM21}
    - \log \left[ u(\bth) \right] \ge \frac{C - u(\btht)}{u(\btht)} \log \left[ C - u(\bth) \right] + \ct{02},
\end{equation}
where $C > u(\btht)$ and $\ct{02}$ is a constant free from $\bth$. The equality in (\ref{eqn2.NEM21}) holds iff $\bth=\btht$. The usefulness of (\ref{eqn2.NEM21}) can be found in Section 3.3.
\section{$\!\!\!\!\!\!\!\!$. Application 1: Problems solved for the first time}     

In this section, we present five models, where their MLEs were not solved by the EM-type algorithms previously. We shall show that we can solve them by the N-EM algorithm for the first time. In other words, the N-EM algorithm can solve some problems that cannot be solved by the EM-type algorithms.

\subsection{Gamma--integral distribution}  

A random variable $X$ is said to follow the gamma--integral distribution with parameter $\th > 0$, denoted by $X \sim \GI(\th)$, if its pdf is (Devroye 1986, p.191)
$$
   f(x| \th) = \int_x^{\infty} \frac{s^{\th -2} \e^{-s}} {\Ga(\th)} \rd s, \quad x>0.
$$
Let $\{X_i\}_{i=1}^n \iid \GI (\th)$ and $\Y{obs} = \{x_i\}_{i=1}^n$ denote the observed data with $x_i$ being the realization of $X_i$, then the log-likelihood function of $\th$ is
\begin{equation}  \label{eqn3.NEMTGL.19}
   \ell(\th | \Y{obs}) = \sum_{i=1}^n \log \left[\int_{x_i}^{\infty} \frac{s^{\th-2} \e^{-s}}{\Ga(\th)} \rd s\right] = \sum_{i=1}^n \log \left[ \int_{\bbS} h(s| x_i, \th) \rd s\right],
\end{equation}
which is of the form of (\ref{eqn2.1}) and $\bbS = (0, \infty) = \bbR_+$, where
$$
   h(s|x_i, \th) \teq  \frac{s^{\th-2} \e^{-s}}{\Ga(\th)}  \cdot I(s > x_i).
$$
By normalizing the $h(\cdot |x_i, \th)$ function, the N-step of the N-EM is to construct the ndf as
\begin{equation}  \label{eqn3.NEMTGL.20}
   g(s | x_i, \th) = \frac{h(s|x_i, \th)}{\int_{\bbS} h(s'|x_i, \th) \rd s'} = \frac{s^{\th-2} \e^{-s} \cdot I(s > x_i)}{\int_{x_i}^{\infty}(s')^{\th-2} \e^{-s'}\rd s'}, \quad s \in \bbS.
\end{equation}
The E-step of the N-EM algorithm is to calculate the $Q$-function as
\begin{eqnarray} \label{eqn3.NEMTGL.21}
   Q(\th | \tht) &=& \sum_{i=1}^n \int_{\bbS} \log \left[ h(s|x_i, \th) \right] \cdot g(s | x_i, \tht) \rd s \non \\ [2mm]
   &=& \sum_{i=1}^n \int_{\bbS} \big\{(\th-2) \log(s) - s - \log [\Ga(\th)] \big\} \cdot  g(s | x_i, \tht) \rd s \non \\ [2mm]
   &=& \ct{} + \th \sum_{i=1}^n a(x_i, \tht) - n \log [\Ga(\th)],
\end{eqnarray}
where $\ct{}$ is a constant free from $\th$, and
$$
   a(x_i, \tht) \teq E_{g}[ \log(S) |x_i, \tht] = \int_{\bbS} \log(s) \cdot g(s | x_i, \tht) \rd s
$$
is the expectation of $\log(S)$ with respect to the density $g(s | x_i, \tht)$. The M-step of the N-EM algorithm is to update $\tht$ with $\th^{(t+1)}$ by solving the following equation:
\begin{equation} \label{eqn3.NEMTGL.22}
   \frac{\rd \Ga(\th)/\rd \th} {\Ga(\th)} =  \frac{1}{n}\sum_{i=1}^n a(x_i, \tht).
\end{equation}

\subsection{Skew symmetric-normal distribution}  

The class of skew symmetric distributions has received much attention in recent years. In this subsection, we focus on parameter estimates in a general family of skew symmetric distributions which are generated by the \textit{cumulative distribution function} (cdf) of the normal distribution. There is no existing algorithm to estimate the parameters of skew-$t$-normal distribution and skew-Laplace-normal distribution, while the derivation of the EM algorithm for skew-normal-normal distribution proposed by Teimouri (2020) is extremely complex. We shall show that the proposed N-EM algorithm is efficient in calculating the MLEs of parameters in the class of skew symmetric distributions.

In the skew symmetric-normal distribution, $\de \in \bbR$ denotes the skewness parameter with $\de > 0$ if the pdf is right skewed and $\de < 0$ if it is left skewed. In this subsection we only consider the case of $\de>0$, while the case of $\de<0$ can be considered similarly. In addition, we use $N(\cdot|\mu,\si^2)$ to denote the pdf the $N(\mu,\si^2)$, i.e.,
$$
   N(x | \mu, \si^2) \teq \frac{1}{\sqrt{2 \pi} \si} \exp \left[-\frac{(x-\mu)^2}{2 \si^2}\right].
$$

\subsubsection{Skew normal--normal distribution}  

A random variable $X$ is said to follow the skew normal (also known as skew normal--normal) distribution with location parameter $\mu \in \bbR$, scale parameter $\si >0$ and skewness parameter $\de\in \bbR$, denoted by $X \sim \SN(\mu, \si^2, \de)$, if its pdf is (Nadarajah \& Kotz 2003, G\'{o}mez \Et 2007)
$$
   f(x| \bth) = \frac{2}{\sqrt{2 \pi} \si} \exp \left[-\frac{(x-\mu)^2} {2 \si^2}\right] \Phi \left(\frac{x-\mu}{\si}\de\right), \quad x \in \bbR,
$$
where $\bth \teq (\mu, \si^2, \de)^{\T}$ and $\Phi(\cdot)$ denotes the cdf of $N(0, 1)$.

Let $\{X_i\}_{i=1}^n \iid \SN (\mu, \si^2, \de)$ with $\de >0$ and $\Y{obs} = \{x_i\}_{i=1}^n$ denote the observed data with $x_i$ being the realization of $X_i$, then the log-likelihood function of $\bth$ is
\begin{eqnarray*}
   \ell_1(\bth | \Y{obs}) &=& c_1 - \frac{n \log (\si^2)}{2}-\sum_{i=1}^n \frac{(x_i-\mu)^2}{2 \si^2}+\sum_{i=1}^n \log \left[\Phi \left(\frac{x_i-\mu}{\si_*} \right)\right]  \\ [2mm]
   &=& c_1 - \frac{n \log (\si^2)}{2}-\sum_{i=1}^n \frac{(x_i-\mu)^2}{2 \si^2}+\sum_{i=1}^n \log \left[ \int_{-\infty}^{x_i} N(s|\mu, \si_*^2) \rd s \right]   \\ [2mm]
   &=& c_1 - \frac{n \log (\si^2)}{2}-\sum_{i=1}^n \frac{(x_i-\mu)^2}{2 \si^2} + \sum_{i=1}^n \log \left[ \int_{\bbS} h_1(s|x_i, \bth) \rd s \right],
\end{eqnarray*}
which is of the form of (\ref{eqn2.1}), where $c_1$ is a constant, $\si_* \teq \si/\de>0$, $\bbS = (-\infty, \infty) = \bbR$, and
\begin{equation}  \label{eqn3.NEMTGL.1}
   h_1(s|x_i, \bth) \teq  N(s |\mu, \si_*^2) \cdot I(s < x_i).
\end{equation}
By normalizing the $h_1(\cdot |x_i, \bth)$ function, the N-step of the N-EM is to construct the ndf as
\begin{equation}  \label{eqn3.NEMTGL.2}
   g_1(s | x_i, \bth) = \frac{h_1(s|x_i, \bth)}{\int_{\bbS} h_1(s'|x_i, \bth) \rd s'} = \frac{N(s |\mu, \si_*^2) \cdot I(s < x_i)}{\Phi((x_i-\mu) / \si_*)}, \quad s \in \bbS,
\end{equation}
which is the pdf of the univariate truncated normal distribution $\TN(\mu, \si_*^2; -\infty, x_i)$. The E-step is to calculate the $Q$-function as
\begin{eqnarray} \label{eqn3.NEMTGL.3}
   Q_1(\bth | \btht) &=& c_1 - \frac{n \log (\si^2)}{2}-\sum_{i=1}^n \frac{(x_i-\mu)^2}{2 \si^2} + \sum_{i=1}^n \int_{\bbS} \log \left[ h_1(s|x_i, \bth) \right] \cdot g_1(s | x_i, \btht) \rd s \non \\ [2mm]
   &=& \ct{11} - \frac{n \log (\si^2)}{2} - \sum_{i=1}^n \frac{(x_i-\mu)^2} {2\si^2}- \frac{n \log (\si_*^2)}{2} - \frac{n\si_*^{2(t)} + n(\mu - \mut)^2} {2\si_*^2} \non \\[2mm]
   & & +\; \sum_{i=1}^n \frac{\si_*^{2(t)}(x_i+\mut-2 \mu) \cdot g_1(x_i |x_i, \btht)}{2\si_*^2},
\end{eqnarray}
where $\ct{11}$ is a constant free from $\bth$, and the M-step is to update
\begin{equation} \label{eqn3.NEMTGL.4}
   \left\{\begin{array}{lcl}
      \mu^{(t+1)}  &=& \dis\frac{n\Bx +n\mut \de^{(t)2}-\si^{2(t)} \sum_{i=1}^n g_1(x_i | x_i, \btht)} {n+n\de^{(t)2}}, \\ [4mm]
      \si^{2(t+1)} &=& \dis\sum_{i=1}^n \frac{(x_i-\mu^{(t+1)})^2} {2n} + \frac{\si^{2(t)} + \de^{(t)2}(\mut-\mu^{(t+1)})^2}{2} - \dis\frac{\tau_1^{(t)}}{2n}, \\ [4mm]
      \de^{(t+1)2} &=& \dis\frac{n\si^{2(t+1)} \de^{(t)2}} {n[\si^{2(t)} + \de^{(t)2} (\mut-\mu^{(t+1)})^2]-\tau_1^{(t)}},
   \end{array} \right.
\end{equation}
where $\Bx =(1/n)\sum_{i=1}^n x_i$ and
\begin{equation} \label{eqn3.NEMTGL.5}
   \tau_1^{(t)} \teq \si^{2(t)}\sum_{i=1}^n(x_i + \mut - 2 \mu^{(t+1)}) \cdot g_1(x_i | x_i, \btht).
\end{equation}

\subsubsection{Skew $\ibt$-normal distribution}  

A random variable $X$ is said to follow a skew $t$-normal distribution with location parameter $\mu \in \bbR$, scale parameter $\si >0$, skewness parameter $\de \in \bbR$, and degrees of freedom $\nu >0$, denoted by $X \sim \StN(\mu, \si^2, \de, \nu)$, if it has the pdf (Nadarajah \& Kotz 2003, G\'{o}mez \Et 2007)
$$
   f(x|\bth, \nu)=\frac{2}{\si} t_{\nu} \left(\frac{x-\mu}{\si} \right) \Phi \left(\frac{x-\mu}{\si} \de \right), \quad x \in \bbR,
$$
where $\bth \teq (\mu, \si^2, \de)^{\T}$ and
$$
   t_{\nu}(x)=\frac{\Ga(\frac{\nu+1}{2})}{\Ga(\frac{\nu}{2}) \sqrt{\nu \pi}} \left( 1+\frac{x^2}{\nu} \right)^{-\frac{\nu+1}{2}}
$$
is the pdf of the standard Student-$t$ distribution with $\nu$ degrees of freedom. In general parameter estimations, the degree of freedom parameter $\nu$ is assumed to be known, the reason can be found in Lucas (1997).

Let $\{X_i\}_{i=1}^n \iid \StN(\mu, \si^2, \de, \nu)$ with $\de >0$ and $\Y{obs} = \{x_i\}_{i=1}^n$, then the log-likelihood function of $\bth$ is
\begin{eqnarray*}
   \ell_2(\bth | \Y{obs}) &=& c_2 - \frac{n \log (\si^2)}{2} -\frac{\nu+1}{2} \sum_{i=1}^n \log \left[1+\frac{(x_i-\mu)^2}{\nu \si^2} \right] + \sum_{i=1}^n \log \left[\Phi \left(\frac{x_i-\mu}{\si} \de\right)\right],
\end{eqnarray*}
where $c_2$ is a constant. In the framework of the N-EM algorithm, we can construct the same $h_1(s|x_i, \bth)$ given by (\ref{eqn3.NEMTGL.1}), the same
ndf $g_1(s|x_i,\bth)$ given by (\ref{eqn3.NEMTGL.2}), and the E-step is to calculate the $Q$-function as
\begin{eqnarray} \label{eqn3.NEMTGL.6}
   & & Q_2(\bth | \btht) \non \\ [2mm]
   &=& c_2 - \frac{n \log (\si^2)}{2} -\frac{\nu+1}{2} \sum_{i=1}^n \log \left[1+\frac{(x_i-\mu)^2}{\nu \si^2} \right] + \sum_{i=1}^n \int_{\bbS} \log \left[ h_1(s|x_i, \bth) \right] \cdot g_1(s | x_i, \btht) \rd s \non \\ [2mm]
   &=& \ct{21} - \frac{n \log (\si^2)}{2} -\frac{\nu+1}{2} \sum_{i=1}^n \log \left[1+\frac{(x_i-\mu)^2}{\nu \si^2} \right] - \frac{n \log (\si_*^2)}{2} - \frac{n\si_*^{2(t)} + n(\mu-\mut)^2}{2\si_*^2} \non \\[2mm]
   & & +\; \sum_{i=1}^n \frac{\si_*^{2(t)}(x_i+\mut-2 \mu) \cdot g_1(x_i |x_i, \btht)}{2\si_*^2}.
\end{eqnarray}
Because of the term $-\sum_{i=1}^n \log \left[1+(x_i-\mu)^2/(\nu \si^2) \right]$, directly maximizing the $Q$-function (\ref{eqn3.NEMTGL.6}) is not available in the sense that there are no closed-form solutions to $\bth^{(t+1)}$.

In (\ref{eqn2.NEM17}), let $\vp(z) = -\log (1+z)$ for $z\in \bbR_+ = (0, \infty)$, we have
\begin{equation} \label{eqn3.NEMTGL.7}
   -\log (1+z) \ge-\log (1+z_0)-\frac{z}{1+z_0}+\frac{z_0}{1+z_0}, \quad \forall \, z>0, \; z_0 >0.
\end{equation}
By setting
$$
   z = \frac{(x_i-\mu)^2}{\nu \si^2}, \quad z_0 = \frac{(x_i-\mut)^2}{\nu \si^{2(t)}},
$$
and plugging them into (\ref{eqn3.NEMTGL.7}), we obtain
$$
   -\log \left[1+\frac{(x_i-\mu)^2}{\nu \si^2}\right] \ge \ct{22} - \frac{(x_i-\mu)^2 /(\nu \si^2)}{1+\left( x_i-\mut \right)^2 /\left(\nu \si^{2(t)}\right)}.
$$
Then, we can construct a new $Q$-function
\begin{eqnarray} \label{eqn3.NEMTGL.8}
   Q_2^{\rm New}(\bth|\btht) &=& \ct{23} - \frac{n \log (\si^2)}{2} - \sum_{i=1}^n  \frac{w_i^{(t)}(x_i-\mu)^2} {2\si^2}  - \frac{n\log (\si_*^2)}{2} - \frac{n\si_*^{2(t)} + n(\mu - \mut)^2}{2 \si_*^2}  \non \\[2mm]
   & & +\; \sum_{i=1}^n \frac{\si_*^{2(t)}(x_i+\mut-2 \mu) g_1(x_i | x_i, \btht)}{2 \si_*^2},
\end{eqnarray}
which is the OLB of the $Q_2(\bth|\btht)$ function (\ref{eqn3.NEMTGL.6}) with tangent point being $\bth=\btht$, where
$$
   w_i^{(t)} \teq \frac{\nu+1} {\nu+\left(x_i-\mut\right)^2 / \si^{2(t)}}, \quad i=1, \ldots, n.
$$
By maximizing the $Q_2^{\rm New}(\bth|\btht)$ function defined in (\ref{eqn3.NEMTGL.8}), we update
\begin{equation} \label{eqn3.NEMTGL.9}
   \left\{\begin{array}{lcl}
      \mu^{(t+1)} &=& \tau_2^{(t)}\left[\dis\sum_{i=1}^n w_i^{(t)} x_i+n \de^{(t)2} \mut-\si^{2(t)} \sum_{i=1}^n g_1(x_i | x_i, \btht)\right], \\[4 mm]
      \si^{2(t+1)} &=& \dis \sum_{i=1}^n \frac{ w_i^{(t)}(x_i-\mu^{(t+1)})^2}{2 n}+\frac{\si^{2(t)}+\de^{(t)2}(\mut-\mu^{(t+1)})^2}{2} -\frac{\tau_1^{(t)}}{2n}, \\[4 mm]
      \de^{(t+1)2} &=& \dis\frac{n\si^{2(t+1)}\de^{(t)2}}{n[\si^{2(t)}+\de^{(t)2}(\mut-\mu^{(t+1)})^2]-\tau_1^{(t)}},
   \end{array}\right.
\end{equation}
where  $\tau_1^{(t)}$ is defined in (\ref{eqn3.NEMTGL.5}) and
\begin{equation}\label{eqn3.NEMTGL.10}
   \tau_2^{(t)} \;\hat{=}\; \left(\sum_{i=1}^n w^{(t)}_i+n \de^{(t)2}\right)^{-1}.
\end{equation}

\subsubsection{Skew Laplace--normal distribution}  

A random variable $X$ is said to follow a skew Laplace--normal distribution with location parameter $\mu \in \bbR$, scale parameter $\si >0$, and skewness parameter $\de \in \bbR$, denoted by $X \sim \SLN(\mu, \si^2, \de)$, if it has the pdf (Nadarajah \& Kotz 2003, G\'{o}mez \Et 2007)
$$
    f(x| \bth) = \frac{1}{\si} \exp \left( - \frac{|x-\mu|}{\si} \right) \Phi\left( \frac{x-\mu}{\si}\de\right), \quad x \in \bbR,
$$
where $\bth \teq (\mu, \si^2, \de)^{\T}$. Let $\{X_i\}_{i=1}^n \iid \SLN(\mu, \si^2, \de)$ with $\de >0$ and $\Y{obs} = \{x_i\}_{i=1}^n$, then the log-likelihood function of $\bth$ is
\begin{eqnarray*}
    \ell_3(\bth | \Y{obs}) = c_3-\frac{n \log (\si^2)}{2} - \sum_{i=1}^n \frac{|x_i-\mu|}{\si} + \sum_{i=1}^n\log \left[\Phi \left(\frac{x_i-\mu}{\si}\de \right)\right],
\end{eqnarray*}
where $c_3$ is a constant. The N-step of the N-EM algorithm is to construct the same $h_1(s|x_i, \bth)$ as (\ref{eqn3.NEMTGL.1}), the same ndf $g_1(s|x_i,\bth)$ as (\ref{eqn3.NEMTGL.2}), and the E-step is to calculate the $Q$-function as
\begin{eqnarray} \label{eqn3.NEMTGL.11}
   Q_3(\bth | \btht)
   &=&  \ct{31} - \frac{n \log (\si^2)}{2} - \sum_{i=1}^n \frac{|x_i-\mu|}{\si}  - \frac{n \log (\si_*^2)}{2} - \frac{n\si_*^{2(t)}+n(\mu - \mut)^2} {2\si_*^2} \non \\[2mm]
   & & +\; \sum_{i=1}^n \left[\frac{\si_*^{2(t)}(x_i+\mut-2 \mu) \cdot g_1(x_i |x_i, \btht)}{2\si_*^2}\right].
\end{eqnarray}
Similarly, because of the term $-\sum_{i=1}^n |x_i-\mu|/\si$, directly maximizing the $Q$-function (\ref{eqn3.NEMTGL.11}) is not available in the sense that there are no closed-form solutions to $\bth^{(t+1)}$.

In (\ref{eqn2.NEM17}), let $\vp(z) = z^2$ for $z \in \bbR = (-\infty, \infty)$, we have
\begin{equation} \label{eqn3.NEMTGL.12}
   z^2 \ge z_0^2 + 2(z - z_0)z_0 \qor -2zz_0 \ge - z^2 - z_0^2, \quad \forall \, z, \; z_0 \in \bbR.
\end{equation}
By setting
$$
   z = \left\{\frac{\sit}{|x_i-\mut|}  \cdot \frac{(x_i-\mu)^2}{2\si^2} \right\}^{1/2} \teq \left\{w_i^{(t)} \cdot \frac{(x_i-\mu)^2}{2\si^2}  \right\}^{1/2}, \quad
   z_0 =  \left(\frac{1}{2 w_i^{(t)}} \right)^{1/2},
$$
and plugging them into (\ref{eqn3.NEMTGL.12}), we obtain
$$
   -\sum_{i=1}^n \frac{|x_i-\mu|}{\si} \ge -\sum_{i=1}^n \frac{w_i^{(t)}(x_i-\mu)^2 } {2\si^2} + \ct{32}.
$$
Then, we can construct a new $Q$-function
\begin{eqnarray} \label{eqn3.NEMTGL.13}
   Q_3^{\rm New}(\bth|\btht) &=& \ct{33} - \frac{n \log (\si^2)}{2} - \sum_{i=1}^n  \frac{w_i^{(t)}(x_i-\mu)^2} {2\si^2}  - \frac{n\log (\si_*^2)}{2} - \frac{n\si_*^{2(t)} + n(\mu - \mut)^2}{2 \si_*^2}  \non \\[2mm]
   & & +\; \sum_{i=1}^n \frac{\si_*^{2(t)}(x_i+\mut-2 \mu) g_1(x_i | x_i, \btht)}{2 \si_*^2},
\end{eqnarray}
which is the OLB of the $Q_3(\bth|\btht)$ function (\ref{eqn3.NEMTGL.11}) with tangent point being $\bth=\btht$. By maximizing the $Q_3^{\rm New}(\bth|\btht)$ function (\ref{eqn3.NEMTGL.13}), we update
\begin{equation} \label{eqn3.NEMTGL3.14}
   \left\{\begin{array} {lcl}
      \mu^{(t+1)} &=& \tau_3^{(t)}\left[\dis\sum_{i=1}^n w_i^{(t)} x_i+n \de^{(t)2} \mut-\si^{2(t)} \sum_{i=1}^n g_1(x_i | x_i, \btht)\right], \\[6mm]
      \si^{2(t+1)} &=& \dis \sum_{i=1}^n \frac{ w_i^{(t)}(x_i-\mu^{(t+1)})^2}{2 n} + \frac{\si^{2(t)} + \de^{(t)2}(\mut-\mu^{(t+1)})^2}{2} - \dis\frac{\tau_1^{(t)}}{2n}, \\[6mm]
      \de^{(t+1)2} &=& \dis\frac{n\si^{2(t+1)}\de^{(t)2}}{n[\si^{2(t)}+\de^{(t)2}(\mut-\mu^{(t+1)})^2]-\tau_1^{(t)}},
    \end{array}\right.
\end{equation}
where $\tau_1^{(t)}$ is defined by ($\ref{eqn3.NEMTGL.5}$),
$$
   w_i^{(t)} \teq \frac{\sit}{|x_i-\mut|}, \quad i=1, \ldots, n, \qand
   \tau_3^{(t)} \teq \left(\sum_{i=1}^n w^{(t)}_i+n \de^{(t)2}\right)^{-1}.
$$

\subsection{Two--sided truncated normal distribution}  

A random variable $X$ is said to follow a normal distribution truncated on the interval $[a, b]$ ($a, b$ are known and $a < b$) with location parameter $\mu \in \bbR$ and scale parameter $\si >0$, denoted by $X \sim \TN(\mu, \si^2; [a,b])$ or $X \sim \TN(\mu, \si^2; a,b)$, if it has the pdf (Cohen 1949)
$$
   f(x | \bth; a, b)=\frac{N(x|\mu, \si^2)}{u(\bth)}  \cdot I(a \le x \le b),
$$
where $\bth \teq (\mu, \si^2)^{\T}$ and
$$
    u(\bth) = \int_a^b N(y|\mu, \si^2) \rd y  =   \Phi \left( \frac{b-\mu}{\si} \right) - \Phi \left( \frac{a-\mu}{\si} \right)
$$
is the normalizing constant.

Let $\{X_i\}_{i=1}^n \iid \TN(\mu, \si^2; [a,b])$ and $\Y{obs} = \{x_i\}_{i=1}^n$, then the log-likelihood function of $\bth$ is
\begin{eqnarray} \label{eqn3.NEMTGL3.15}
   \ell_4(\bth | \Y{obs}) = c_4 - \frac{n \log (\si^2)}{2} - \sum_{i=1}^n \frac{(x_i-\mu)^2}{2 \si^2} - n \log [ u(\bth)],
\end{eqnarray}
which is not of the form of (\ref{eqn2.1}). By applying (\ref{eqn2.NEM21}) with $C = 1$, we have
\begin{eqnarray*}
    -\log [ u(\bth)] &\srge{(\ref{eqn2.NEM21})} & \frac{1-u(\btht)} {u(\btht)} \log[1- u(\bth)] + \ct{41} \\ [2mm]
    & = & \st_1 \log\left[ \Phi \left(\frac{a-\mu}{\si} \right) + \Phi \left(\frac{-b+\mu}{\si}\right) \right]  + \ct{41} \\ [2mm]
    &\srge{(\ref{eqn2.NEM19})} & \st_1 \st_2 \log
    \left[\Phi\left(\frac{ a-\mu}{\si}\right)\right] + \st_1 (1-\st_2) \log \left[\Phi\left(\frac{-b+\mu}{\si}\right)\right] + \ct{42},
\end{eqnarray*}
where
\begin{eqnarray*}
   u(\btht) &=& \Phi \left(\frac{b - \mut}{\sit}\right) - \Phi \left(\frac{a-\mut}{\sit}\right), \\ [2mm]
   1-u(\btht) &=& \Phi \left(\frac{a-\mut}{\sit}\right) + \Phi \left(\frac{-b + \mut}{\sit}\right), \\ [2mm]
   \st_1 &\teq & \frac{1 - u(\btht)} {u(\btht)}, \qquad
   \st_2 \;\teq\; \frac{\Phi\left(\frac{a-\mut}{\sit}\right)} {1-u(\btht)}.
\end{eqnarray*}
So, we can construct a $Q$-function
\begin{eqnarray} \label{eqn3.NEMTGL3.16}
   Q_4(\bth | \btht) &=& \ct{43}-\frac{n \log (\si^2)}{2} - \sum_{i=1}^n \frac{(x_i-\mu)^2}{2\si^2}+ n \st_1 \st_2  \log \left[\Phi \left(\frac{a-\mu}{\si}\right)\right] \non \\ [2mm]
   & & +\; n s_1^{(t)} (1-\st_2 ) \log \left[\Phi \left(\frac{-b+\mu}{\si}\right) \right].
\end{eqnarray}

In the N-EM framework, we can construct the same ndf $g_1(s|x_i,\bth)$ with $\de=1$ as (\ref{eqn3.NEMTGL.2}), and the E-step is to calculate a new $Q$-function as
\begin{eqnarray} \label{eqn3.NEMTGL3.17}
   Q_4^{\rm New}(\bth | \btht) &=& \ct{44} - \frac{n(1+s_1 ^{(t)}) \log (\si^2)}{2}-\sum_{i=1}^n \frac{(x_i-\mu)^2}{2 \si^2} -n \st_1  \frac{\si^{2(t)}+(\mu-\mut)^2}{2 \si^2} \non \\ [2mm]
   & & +\; n\st_1 \st_2 \frac{\si^{2(t)} (a+ \mut -2\mu) g_1(a|a, \btht)} {2 \si^2}   \non \\ [2mm]
   & & -\; n \st_1 (1-\st_2 ) \frac{\si^{2(t)}(\mut-2 \mu+b) g_1(-b | -b, \btht)}{2 \si^2},
\end{eqnarray}
which is the OLB of the $Q_4(\bth|\btht)$ function (\ref{eqn3.NEMTGL3.16}) with tangent point being $\bth=\btht$. By maximizing the $Q_4^{\rm New}(\bth|\btht)$ function (\ref{eqn3.NEMTGL3.17}), we update
\begin{equation} \label{eqn3.NEMTGL3.18}
   \left\{ \begin{array}{lcl}
      \mu^{(t+1)} &=& \dis\frac{\Bx + \st_1 \left\{\mut - \si^{2(t)} \left[\st_2 g_1(a | a, \btht)-(1-\st_2 )g_1(-b | -b, \btht)\right]\right\}}{1 + \st_1}, \\[6mm]
      \si^{2(t+1)} &=& \dis\frac{ \sum_{i=1}^n (x_i- \mu^{(t+1)} )^2 + n \st_1 \si^{2(t)} \tau_4^{(t)}}{n(1 + \st_1)},
   \end{array} \right.
\end{equation}
where
\begin{eqnarray*}
   \tau_4^{(t)} &\teq& 1 +  \frac{(\mut-\mu^{(t+1)})^2} {\si^{2(t)}} - (a+\mut-2\mu^{(t+1)})\st_2  g_1(a | a, \btht) \\[2mm]
   & & +\; (\mut-2\mu^{(t+1)} + b)(1-\st_2 )g_1(-b | -b, \btht).
\end{eqnarray*}
\section{$\!\!\!\!\!\!\!\!$. Application 2: Three kinds of existing EM algorithms \\ \hspace*{-0.38cm} are unified in the N-EM framework}

The most difficult issue associated with the EM-type algorithms is the introduction of latent variables/vectors or the data augmentation (Wang \& Tanner 1987). For example, although EM-type algorithms can be applied to mixture problems, \textit{stochastic representation} (SR) methods and group problems, designing such EM-type algorithms is often case by case. In this section, we shall illustrate that these EM-type algorithms can be unified in the proposed N-EM framework, avoiding the data augmentation and hence resulting in a quite straightforward algorithm construction.

\subsection{Mixture problems}  

Many distributions or models can be generated by a basic distribution mixing with other distribution(s), for example, the Student-$t$ distribution is a normal mixture with a gamma distribution. The EM algorithm can be constructed for finding MLEs of parameters in mixture models with mixing random variables as the latent variables. However, the N-EM algorithm can provide a unified framework to systematically deal with the mixture problems. In this subsection, we present three examples about the multivariate Laplace distribution and its extensions.

\subsubsection{Multivariate Laplace distribution}  

As a multivariate normal mixture with the standard exponential as the mixing distribution, the multivariate (symmetric) Laplace is also an important distribution in statistics just like the multivariate (symmetric) $t$-distribution. In this paper, the pdf of $U \sim \Exponential(\be)$ is denoted by $\be \e^{-\be u}\cdot I(u > 0)$. Let
\begin{equation}   \label{eqn3.NEMTGL3.23}
   U\sim \Exponential(1)  \qand \bx |U \sim N_d(\bmu, 2U \bSi),
\end{equation}
then $\bx = (X_1,\ldots, X_d)^{\T}$ is said to follow the $d$-dimensional (symmetric) Laplace distribution, denoted by $\bx \sim \Laplace_d(\bmu,\bSi)$, where $\bmu = (\mu_1, \ldots, \mu_d)^{\T} \in \bbR^d$ is the mean parameter vector and $\bSi=(\si_{jj'}) >0$ is the scale parameter matrix. From the mixture representation (\ref{eqn3.NEMTGL3.23}), we obtain an equivalent SR: $\bx = \bmu + \sqrt{2U}\by$, where $\by \sim N_d(\0, \bSi)$ and $U \inde \by$. From this SR, it is easy to obtain $E(\bx) = \bmu$ and $\Var(\bx) = 2\bSi$. Let $\bth=\{\bmu, \bSi\}$, then the pdf of $\bx$ is
$$
   f_{\bx}(\ibx|\bth) = (4\pi)^{-d/2}|\bSi|^{-1/2}
   \int_0^\infty  h_1(u|\ibx, \bth) \rd u,
$$
where
$$
   h_1(u|\ibx, \bth) = u^{-d/2} \exp \left[-\frac{(\ibx-\bmu)^{\T} \bSi^{-1}(\ibx-\bmu)}{4u} - u \right].
$$

Let $\{\bx_i\}_{i=1}^n \iid \Laplace_d(\bmu,\bSi)$ and $\Y{obs} = \{\ibx_i\}_{i=1}^n$ denote the observed data with $\ibx_i$ being the realization of $\bx_i$, then the log-likelihood function of $\bth$ is given by
$$
   \ell_1(\bth|\Y{obs}) = c_1 -\frac{n}{2}\log|\bSi|+ \sum_{i=1}^n \log\left[\int_0^\infty h_1(u|\ibx_i,\bth)\rd u \right],
$$
which is of the form of (\ref{eqn2.1}). By normalizing the $h_1(u|\ibx_i,\bth)$ function, the N-step of the N-EM algorithm is to construct the ndf as
$$
   g_1(u|\ibx_i, \bth) = \frac{h_1(u|\ibx_i, \bth)} {\int_0^\infty h_1(u'|\ibx_i,\bth) \rd u'}, \quad u>0,
$$
the E-step is to calcualte the $Q$-function as
\begin{eqnarray} \label{eqn3.NEMTGL3.24}
    Q_1(\bth|\btht) = \ct{11} -\frac{n}{2} \log|\bSi| - \sum_{i=1}^n \frac{(\ibx_i -\bmu)^{\T} \bSi^{-1} (\ibx_i - \bmu)}{4} a_1(\ibx_i, \btht),
\end{eqnarray}
where
$$
   a_1(\ibx_i, \btht) \teq E_{g_1}(U^{-1} |\ibx_i, \btht) = \int_0^\infty u^{-1} \cdot g_1(u|\ibx_i,\btht) \rd u,
$$
and the M-step is to update
\begin{equation} \label{eqn3.NEMTGL3.25}
   \left\{ \begin{array}{lcl}
      \bmu^{(t+1)} &=& \left[ \dis\sum_{i = 1}^n  a_1(\ibx_i, \btht) \right]^{-1} \dis\sum_{i = 1}^n  a_1(\ibx_i, \btht)\ibx_i, \\ [6mm]
      \bSi^{(t+1)} &=& \dis\frac{1}{2n} \sum_{i = 1}^n a_1(\ibx_i, \btht) (\ibx_i -\bmu^{(t+1)})(\ibx_i -\bmu^{(t+1)})^{\T}.
   \end{array} \right.
\end{equation}

\subsubsection{Multivariate asymmetric Laplace distribution}  

The multivariate Laplace distribution has been extended to multivariate asymmetric Laplace distribution. Let
\begin{equation}   \label{eqn3.NEMTGL3.26}
   U\sim \Exponential(1)  \qand \bx |U \sim N_d(\bmu + U\bmu^*, 2U \bSi),
\end{equation}
then $\bx = (X_1,\ldots, X_d)^{\T}$ is said to follow the $d$-dimensional  asymmetric Laplace distribution, denoted by $\bx \sim \ALaplace_d(\bmu, \bmu^*, \bSi)$, where $\bmu = (\mu_1, \ldots, \mu_d)^{\T} \in \bbR^d$ is the symmetric  location parametric vector, $\bmu^* = (\mu_1^*, \ldots, \mu_d^*)^{\T} \in \bbR^d$ is the asymmetric location parameter vector, and $\bSi=(\si_{jj'}) >0$ is the scale parameter matrix. From the mixture representation (\ref{eqn3.NEMTGL3.26}), we obtain an equivalent SR: $\bx = \bmu + U\bmu^* + \sqrt{2U}\by$, where $\by \sim N_d(\0, \bSi)$ and $U \inde \by$. From this SR, it is easy to obtain
$E(\bx) = \bmu + \bmu^*$ and $\Var(\bx) = \bmu^{*\T} \bmu^* + 2\bSi$. Let $\bth=\{\bmu, \bmu^*, \bSi\}$, then the pdf of $\bx$ is
$$
   f_{\bx}(\ibx|\bth) = (4\pi)^{-d/2}|\bSi|^{-1/2} \int_0^\infty h_2(u|\ibx, \bth) \rd u,
$$
where
$$
   h_2(u|\ibx, \bth) = u^{-d/2} \exp\left[-\frac{(\ibx - \bmu - u \bmu^*)^{\T} \bSi^{-1} (\ibx - \bmu - u \bmu^*)}{4u} - u \right].
$$

Let $\{\bx_i\}_{i=1}^n \iid \ALaplace_d(\bmu,\bmu^*,\bSi)$ and $\Y{obs} = \{\ibx_i\}_{i=1}^n$ denote the observed data with $\ibx_i$ being the realization of $\bx_i$, then the log-likelihood function of $\bth$ is given by
$$
   \ell_2(\bth|\Y{obs}) = c_2 -\frac{n}{2}\log|\bSi|+ \sum_{i=1}^n \log\left[\int_0^\infty h_2(u|\ibx_i,\bth)\rd u \right],
$$
which is of the form of (\ref{eqn2.1}). By normalizing the $h_2(u|\ibx_i,\bth)$ function, the N-step of the N-EM algorithm is to construct the ndf as
$$
   g_2(u|\ibx_i, \bth) = \frac{h_2(u|\ibx_i, \bth)} {\int_0^\infty h_2(u'|\ibx_i,\bth) \rd u'}, \quad u>0.
$$
The E-step is to calculate the $Q$-function as
\begin{eqnarray} \label{eqn3.NEMTGL3.27}
   & & Q_2(\bth|\btht) \non \\[2mm]
   &=& c_2 - \frac{n}{2}\log|\bSi| + \sum_{i=1}^n \int_0^\infty \log\left[ h_2(u|\ibx_i, \bth) \right] \cdot g_2(u|\ibx_i, \btht) \rd u \non \\[2mm]
   &=& \ct{21} -\frac{n}{2}\log|\bSi| - \sum_{i=1}^n \int_0^\infty \frac{(\ibx_i - \bmu - u \bmu^*)^{\T} \bSi^{-1} (\ibx_i - \bmu - u \bmu^*)}{4u}
      \cdot g_2(u|\ibx_i,\btht) \rd u \non \\[2mm]
   &=& \ct{21} -\frac{n}{2} \log|\bSi| - \frac{1}{4} \sum_{i=1}^n  \Big[ (\ibx_i - \bmu)^{\T} \bSi^{-1} (\ibx_i -\bmu) \cdot a_2(\ibx_i, \btht)  \non \\[2mm]
   & & +\; \bmu^{*\T}\bSi^{-1}\bmu^* \cdot b_2(\ibx_i, \btht) -2 \bmu^{*\T} \bSi^{-1} (\ibx_i -\bmu) \Big],
\end{eqnarray}
where
\begin{eqnarray*}
   a_2(\ibx_i, \btht) &\teq & E_{g_2}(U^{-1} |\ibx_i, \btht) = \int_0^{\infty} u^{-1} \cdot g_2(u|\ibx_i, \btht)\rd u  \qand \\[2mm]
   b_2(\ibx_i, \btht) &\teq & E_{g_2}(U |\ibx_i, \btht) = \int_0^{\infty} u \cdot g_2(u | \ibx_i, \btht)\rd u.
\end{eqnarray*}
The M-step is to update
\begin{equation} \label{eqn3.NEMTGL3.28}
   \left\{ \begin{array}{lcl}
      \bmu^{(t+1)} &=& \left[\dis\sum_{i=1}^n a_2(\ibx_i, \btht) \right]^{-1} \left[\dis\sum_{i=1}^n a_2(\ibx_i, \btht) \ibx_i - n\bmu^{*(t)} \right] \\ [4mm]
      \bmu^{*(t+1)}&=& \left[\dis\sum_{i = 1}^n b_2(\ibx_i, \btht) \right]^{-1} \left( \dis\sum_{i=1}^n \ibx_i- n\bmut \right), \\ [4mm]
      \bSi^{(t+1)} &=& \dis\frac{1}{2n} \sum_{i = 1}^n \Big[a_2(\ibx_i, \btht) (\ibx_i -\bmut) (\ibx_i-\bmut)^{\T} - (\ibx_i-\bmut) (\bmu^{*(t)})^{\T} \\ [4mm]
      & & -\; \bmu^{*(t)} (\ibx_i-\bmut)^{\T} + b_2(\ibx_i, \btht) \bmu^{*(t)} (\bmu^{*(t)})^{\T} \Big].
   \end{array} \right.
\end{equation}

\subsubsection{Type \II \ multivariate Laplace distribution}  

The multivariate Laplace distribution has been extended to Type \II \ multivariate Laplace distribution, whose random components could have different value for its own mixing variate and are correlated only through the
dependence structure of the normal random vector. Let
\begin{eqnarray} \label{eqn3.NEMTGL3.29}
   U_1, \ldots, U_d \iid \Exp(1) \qand  \bx \mid (\bu = \ibu) \sim N_{d}\left(\bmu, 2 \ibU^{1 / 2} \bSi \ibU^{1 / 2}\right)
\end{eqnarray}
where $\ibu= (u_1, \ldots, u_d)^{\T}$ is the realization of $\bu = (U_1, \ldots, U_d)^{\T}$,
$$
   \ibU^{1/2} = \diag(\sqrt{u_1}, \ldots, \sqrt{u_d}) = \diag(\ibu^{1/2} )
$$
is the realization of $\bU^{1/2} =\diag(\sqrt{U_1}, \ldots, \sqrt{U_d}) = \diag (\bu^{1 / 2})$, then $\bx = (X_1,\ldots, X_d)^{\T}$ is said to follow the $d$-dimensional  Type \II \ Laplace distribution, denoted by  $\bx \sim \Laplace_d^{(\rm \II)}(\bmu, \bSi)$, and $\bmu = (\mu_1, \ldots, \mu_d)^{\T} \in \bbR^d$ is the location parametric vector, and $\bSi=(\si_{jj'}) >0$ is the scale parameter matrix. From the mixture representation (\ref{eqn3.NEMTGL3.29}), we obtain an equivalent SR:
\begin{equation*}
   \bx= \threev{\mu_1}{\vdots}{\mu_d} + \sqrt{2} \threethree{\sqrt{U_1}}{\cdots}{0} {\vdots}{\ddots}{\vdots} {0}{\cdots}{\sqrt{U_d}} \threev{Y_1}{\vdots}{Y_d}
   = \bmu + \sqrt{2} \bU^{1/2} \by,
\end{equation*}
where $\by=(Y_1, \ldots, Y_d)^{\T} \sim N_d(\0, \bSi)$ and $\bu=(U_1, \ldots, U_d)^{\T} \inde \by$. From this SR, it is easy to obtain $E(\bx) = \bmu$ and $\Var(\bx) =  2\bSi$. Let $\bth=\{\bmu,  \bSi\}$, then the pdf of $\bx$ is
$$
   f_{\bx}(\ibx|\bth) =  (4\pi)^{-d/2}|\bSi|^{-1/2} \int_{\bbR_+^d} h_3(\ibu | \ibx, \bth) \rd \ibu,
$$
where
$$
   h_3(\ibu | \ibx, \bth) \teq \left( \prod_{j=1}^d u_j^{-\frac{1}{2}} \right) \exp \left[-\frac{(\ibu^{-\frac{1}{2}})^{\T} (\ibX^* \bSi^{-1}\ibX^*) \ibu^{-\frac{1}{2}}}{4} - \1^{\T}\ibu \right] \qand \ibX^* \teq \diag(\ibx-\bmu).
$$

Let $\{\bx_i\}_{i=1}^n \iid \Laplace_d^{(\rm \II)}(\bmu, \bSi)$ and $\Y{obs}= \{\ibx_i\}_{i=1}^n$ with $\ibx_i$ being the realization of $\bx_i$, then the log-likelihood function of $\bth$ is given by
$$
   \ell_3(\bth|\Y{obs}) = c_3  -\frac{n}{2}\log|\bSi|+ \sum_{i=1}^n \log\left[\int_{\bbR_+^d} h_3(\ibu | \ibx_i, \bth)\rd \ibu\right].
$$
where
which is of the form of (\ref{eqn2.1}). By normalizing the $h_3(\cdot |\ibx_i,\bth)$ function, the N-step of the N-EM algorithm is to construct the ndf as
$$
   g_3(\ibu|\ibx_i, \bth) = \frac{ h_3(\ibu | \ibx_i, \bth)} {\int_{\bbR_+^d}  h_3(\ibu' | \ibx_i, \bth)\rd \ibu'}, \quad \ibu \in \bbR_+^d.
$$
The E-step of the N-EM algorithm is to calculate the $Q$-function as
\begin{eqnarray} \label{eqn3.2.30}
   & & Q_3(\bth|\btht) \non \\[2mm]
   &\teq& c_3  - \frac{n}{2}\log|\bSi|+ \sum_{i=1}^n \int_{\bbR_+^d} \log[ h_3(\ibu | \ibx_i, \bth)] \cdot g_3(\ibu|\ibx_i, \btht)  \rd \ibu \non \\[2mm]
   &=& c_{31}^{(t)} - \frac{n}{2}\log|\bSi| - \sum_{i=1}^n \int_{\bbR_+^d} \left[\frac{(\ibu^{-1/2})^{\T}(\ibX_i^* \bSi^{-1}\ibX_i^*)\ibu^{-1/2}}{4} \right] g_3(\ibu|\ibx_i,\btht) \rd \ibu. \qquad
\end{eqnarray}
The M-step is to update
\begin{equation} \label{eqn3.2.31}
   \left\{ \begin{array}{lcl}
      \bmu^{(t+1)} &=& \left[\dis \sum_{i=1}^n  \ibU_i^{-\frac{1}{2}} \bSi^{-1}\ibU_i^{-\frac{1}{2}} \right]^{-1} \left[\dis\sum_{i = 1}^n  \ibU_i^{-\frac{1}{2}}\bSi^{-1}\ibU_i^{-\frac{1}{2}}\ibx_i\right], \\ [6mm]
      \bSi^{(t+1)} &=& \dis \frac{1}{2 n} \sum_{i=1}^n \ibU_i^{-\frac{1}{2}} (\ibx_i -\bmu) (\ibx_i -\bmu)^{\T} \ibU_i^{-\frac{1}{2}},
   \end{array} \right.
\end{equation}
while replacing $u_{ij}^{-\frac{1}{2}}u_{ik}^{-\frac{1}{2}}$ in (\ref{eqn3.2.31}) by $\int_{\bbR_+^d} u_{ij}^{-\frac{1}{2}} u_{ik}^{-\frac{1}{2}} \cdot g(\ibu|\ibx_i,\btht) \rd \ibu$, where
$$
   \ibU_i^{-\frac{1}{2}} = \diag \left(u_{i1}^{-\frac{1}{2}}, \ldots, u_{id}^{-\frac{1}{2}}\right).
$$

\subsection{Stochastic representation methods}  

A positive random variable $Y$ is said to follow the \textit{inverse Gaussian} (IG) distribution with location parameter $\mu \, (> 0)$ and shape parameter $\la \, (>0)$, if its pdf is
$$
   \f{IG}(y| \mu, \la) = \sqrt{\frac{\la}{2\pi}} \; y^{-\frac{3}{2}} \exp \left[ -\frac{\la(y-\mu)^2} {2\mu^2 y} \right], \quad y>0,
$$
denoted by $Y \sim \IG(\mu, \la)$. Tweedie (1957) obtained $E(Y)=\mu$ and $\Var(Y)=\mu^3/\la$. Especially, when $\la=\mu^2$, the $\IG(\mu, \la)$ reduces to the equal-dispersion IG distribution $\IG(\mu, \mu^2)$ with $E(Y)=\Var(Y)$.

By using two independent IG variates with equal-dispersion, Liu \Et (2021) proposed a \textit{proportional inverse Gaussian} (PIG) distribution. Let $Y_k \sim \IG(\th_k, \th_k^2)$ for $k=1,2$ and $Y_1 \inde Y_2$. Define
\begin{equation} \label{eqn3.2.35}
   X = \frac{Y_1}{Y_1 +Y_2},
\end{equation}
then $X$ is said to follow the PIG distribution with parameters $\th_1 >0$ and $\th _2 >0$, denoted by $X \sim \PIG(\bth)$ with $\bth=(\th_1, \th_2 )^{\T}$. Its pdf is
$$
   \f{PIG}(x | \bth) = \frac{\th_1\th_2}{2\pi} \e^{\th_1 +\th_2} [x(1-x)]^{-\frac{3}{2}} \int_0^{\infty} h_4(s | x, \bth) \rd s, \quad 0<x<1,
$$
where
$$
   h_4(s | x, \bth)  \teq s^{-2} \exp \left\{-\frac{1}{2} \left[s + \left(\frac{\th_1^2}{x} + \frac{\th_2^2}{1-x} \right) \frac{1}{s} \right] \right\}.
$$

Let $\{X_i\}_{i=1}^n \iid \PIG (\bth)$ and $\Y{obs} = \{x_i\}_{i=1}^n$ with $x_i$ being the realization of $X_i$. To calculate the MLEs of $\bth$, one can try to create an EM algorithm by treating $(Y_1, Y_2)$ as latent variables via the SR (\ref{eqn3.2.35}). However, it is not an easy task to derive the conditional predictive distributions of $Y_k |(X=x)$ for $k=1,2$. To avoid this difficulty, Liu \Et (2021) turned to develop an MM algorithm, which is, in fact, an N-EM algorithm.

Note that the observed-data log-likelihood function of the parameter vector $\bth$ is
$$
   \ell_4(\bth | \Y{obs}) = c_4 + n \log(\th_1) + n\log(\th_2) + n(\th_1 + \th_2) + \sum_{i=1}^n \log \left[\int\nolimits_0^{\infty} h_4(s | x_i, \bth) \rd s \right],
$$
where $c_4$ is a constant free from $\bth$. By normalizing the $h_4(\cdot |x_i, \bth)$ function, the N-step of the N-EM algorithm is to construct the ndf as
$$
   g_4(s| x_i, \bth) = \frac{h_4(s | x_i, \bth)}{\int_0^{\infty} h_4(s' | x_i, \bth) \rd s'}, \quad s>0.
$$
The E-step of the N-EM algorithm is to calculate the $Q$-function as
\begin{equation} \label{eqn3.2.36}
   Q_4(\bth | \btht) = \ct{41} + n \log (\th_1 )+n \log (\th_2 )+n(\th_1 +\th_2 )  - \frac{\th_1^2 d_1(\btht) + \th_2^2 d_2(\btht)}{2}.
\end{equation}
where $\ct{41}$ is a constant free from $\bth$,
\begin{eqnarray*}
   d_1(\btht) &=& \sum_{i=1}^n \frac{b_4(x_i, \btht)}{x_i}, \qquad d_2(\btht) \;=\; \sum_{i=1}^n \frac{b_4(x_i, \btht)}{1-x_i} \qand \\ [2mm]
   b_4(x_i, \btht) &\teq &  E_{g_4}(S^{-1} | x_i, \btht ) \;=\; \int_0^{\infty} \frac{g_4(s | x_i, \btht)}{s} \rd s.
\end{eqnarray*}
The M-step is to update
\begin{eqnarray} \label{eqn3.2.37}
   \th_k^{(t+1)} = \frac{n+ \sqrt{n^2+4 n d_k(\btht) }} {2 d_k(\btht) }, \quad k=1,2.
\end{eqnarray}

\subsection{Group problems}  

Let $X$ be a continuous random variable with density function $f(x|\bth)$. There are $n \,(=\sum_{j=1}^r n_j)$ independent observations from $X$, but individual observations $\{X_i\}_{i=1}^n$ are not recorded and only the numbers $\{n_j\}_{j=1}^r$ falling in the mutually exclusive intervals $(t_{j-1}, t_j]$ for $j=1, \ldots, r$ are recorded and $t_0, t_r$ could be $-\infty$ and $\infty$. Thus, the observed data are $\Y{obs}=\{n_j\}_{j=1}^r \cup \{t_0, t_1, \ldots, t_r\}$ and the observed-data log-likelihood function is
$$
   \ell_5(\bth | \Y{obs })= c_5 + \sum_{j=1}^r n_j \log \left[p_j(\bth) \right]
   = c_5 + \sum_{j=1}^r n_j \log \left[\int_{-\infty}^{\infty} h_j(x|\bth; t_{j-1}, t_j) \rd x \right],
$$
where $c_5$ is a constant free from $\bth$, $p_j(\bth)= \Pr(t_{j-1}<X \le t_j) = \int_{t_{j-1}}^{t_j} f(x|\bth) \rd x$, and
$$
   h_j(x |\bth; t_{j-1}, t_j) = f(x|\bth) \cdot I(t_{j-1} < x \le t_j), \quad j=1, \ldots, r.
$$

By normalizing the $h_j(\cdot |\bth; t_{j-1}, t_j)$ function, the N-step of the N-EM algorithm is to construct the ndf as
$$
   g_j(x|\bth; t_{j-1}, t_j) = \frac{h_j(x|\bth; t_{j-1}, t_j)} {\int_{-\infty}^{\infty} h_j(y|\bth; t_{j-1}, t_j) \rd y} = \frac{f(x|\bth) \cdot I(t_{j-1} < x \le t_j)} {\int_{t_{j-1}}^{t_j} f(y|\bth) \rd y }, \quad  x\in \bbR,
$$
which is the original density $f(x|\bth)$ truncated on the interval $(t_{j-1},  t_j]$. The E-step of the N-EM algorithm is to calculate the $Q$-function as
$$
   Q(\bth |\btht)= \sum_{j=1}^r n_j Q_j (\bth |\btht),
$$
where
$$
   Q_j (\bth |\btht) = \int_{-\infty}^{\infty} \log [h_j(x|\bth; t_{j-1}, t_j)] \cdot g_j (x |\btht; t_{j-1}, t_j) \rd x.
$$

Furthermore, let $\bth=(\al, \be)^{\T}$ and
$$
   f(x | \bth) =  \frac{\be^{\al}}{\Ga(\al)}  x^{\al-1} \e^{-\be x}, \quad \al>0, \quad \be>0, \quad x \in \bbR_+,
$$
we have $\log [f(x |\bth)] = \al \log (\be) - \log [\Ga(\al)] + (\al - 1) \log(x) - \be x$. Thus,
\begin{eqnarray*}
   Q_j (\bth |\btht) &=& \int_{-\infty}^{\infty} \Big\{ \al \log (\be) - \log [\Ga(\al)] + (\al - 1) \log(x) - \be x \Big\} \cdot g_j (x |\btht;  t_{j-1}, t_j) \rd x  \\ [2mm]
   &=& \al \log(\be) - \log[\Ga(\al)] + (\al - 1) a_j(\btht) - \be b_j(\btht),
\end{eqnarray*}
where
\begin{eqnarray} \label{eqn3.2.32}
   a_j(\btht) &\teq& E_{g_j}[\log(X) | \btht; t_{j-1}, t_j] \;=  \int_{-\infty}^{\infty} \log(x)  \cdot g_j (x |\btht; t_{j-1}, t_j) \rd x, \non \\ [2mm]
   b_j(\btht) &\teq& E_{g_j} (X | \btht; t_{j-1}, t_j) \;=   \int_{-\infty}^{\infty} x \cdot g_j (x|\btht; t_{j-1}, t_j) \rd x, \qand \non \\ [2mm]
   g_j(x |\bth; t_{j-1}, t_j) &=& \frac{x^{\al-1} \e^{-\be x} } {  \int_{t_{j-1}}^{t_j } y^{\al-1} \e^{-\be y} \rd y} \cdot I(t_{j-1} < x \le t_j), \quad t_0 \teq 0, \quad t_r\teq \infty.
 \end{eqnarray}
Therefore, the $Q$-function is
$$
   Q(\bth |\btht)  =  n \al \log(\be) - n\log[\Ga(\al)] + (\al - 1) \sum_{j=1}^r n_j  a_j(\btht) - \be \sum_{j=1}^r n_j  b_j(\btht),
$$
and the M-step is to updates
\begin{equation} \label{eqn3.2.34}
   \frac{\rd \Ga(\al)/ \rd \al} {\Ga(\al)} \bigg|_{\al=\al^{(t+1)}} = \log (\be^{(t)}) + \frac{1}{n}\sum_{j=1}^r n_j  a_j(\btht) \qand
   \be^{(t+1)} = \frac{n\al^{(t+1)}} {\sum_{j=1}^r  n_j  b_j(\btht)}.
\end{equation}
\section{$\!\!\!\!\!\!\!\!$. Numerical experiments}  

We conduct numerical experiments to assess the practical performances of the proposed N-EM algorithms for examples discussed in Section 4. The simulation was coded in \textsf{R} and run on a desktop in Intel(\textsf{R}) Core(TM) i9-9980 with CPU 2.40 GHz, and the stopping criterion is set to be
$$
    \frac{|\ell(\bth^{(t+1)}| \Y{obs}) - \ell(\btht | \Y{obs} ) |}{|\ell(\btht | \Y{obs} )|} < 10^{-6}.
$$
The first three examples in Section 4.1 includes skew-normal-normal distribution (SNND), skew-$t$-normal distribution (STND), two-side truncted normal distribution (TTND). The next five examples in Section 4.2 are multivariate Laplace-family distribution (consisting of multivariate Laplace distribution (MLD), multivariate asymmetric Laplace distribution (MALD), Type \II \ multivariate Laplace distribution (Type \II \ MLD), linear mixed model (LMM) and group data problems of the exponential distribution (GPED). The last example in Section 4.3 is the folded normal distribution (FND).

We generate $R$ replications from various parameter settings as well as sample size and execute the N-EM algorithm we proposed. The sample size ($n$), the average number of iteration ($K$), the average time of running in seconds (Time), the final likelihood values ($L$), the mean squared error (MSE) and the coverage probability based on the percentile bootstrap CI (CP) are summarized in Table 1, 2, 3 for the examples in Section 4.1, 4.2 and 4.3.

The MSE \RED{is} be defined as
$$
    \frac{1}{R}\sum_{r=1}^R \frac{||\hat{\bth}^{(r)}-\bth||^2}{q},
$$
where $\bth$ denotes the true value and $q$ is the dimension of $\bth$.

As a useful tool, the bootstrap method for calculating a bootstrap \textit{confidence interval} (CI) of an arbitrary function of $\bth$, say $\vth = h(\bth)$, is given as follows:
\begin{itemize}
  \item[\hspace*{-0.03cm} Step 1:] We calculate the MLE of $\bth$, denoted by $\hbth$, which can be obtained via the N-EM algorithm. Thus, $\hvth = h(\hbth)$ is the MLE of $\vth$.

  \item[\hspace*{-0.03cm} Step 2:] Based on $\hbth$, we can generate R independent identically distributed replications  $W_1^* = w_1^*, \ldots , W_R^* = w_R^* $. Then we can calculate $\hbth^*$ based on $\Y{obs}^*= \{w_1^*, \ldots, w_R^* \}$ through the N-EM algorithm as in Step 1 and obtain a bootstrap replication $\hvth^* = h(\hbth^*)$.

  \item[\hspace*{-0.03cm} Step 3:] Independently repeating Step 2 $G$ times, we can obtain $G$ bootstrap replications $\{\hvth_g^*\}_{g=1}^G$.

  \item[\hspace*{-0.03cm} Step 4:] The $100(1-\al)\%$ bootstrap CI of $\hvth$ is given by $[\hvth_L, \; \hvth_U]$, where $\hvth_L$ and $\hvth_U$ are the $100(\al/2)$ and $100(1-\al/2)$ percentiles of $\{\hvth_g^*\}^G_{g=1}$, respectively.
\end{itemize}

Here we introduce the parameters configurations in the numerical experiments. For the SNND example, we choose $\bth_1=(\mu, \si, \de)=(3,0.5,2)^{\T}$, $n \in \{1000,10000,50000\}$. For the STND example, we choose $\bth_2 =(\mu, \si, \de, \nu )=(3,0.5,2,2)^{\T}$, $n \in \{1000,10000,50000\}$. For the TTND example, we choose $\bth_3=(\mu, \si )=(3,0.5)^{\T}$, $n \in \{500,1000,5000,10000,50000\}$.

For the Multivariate Laplace-family distribution, we set the parameters, relatively. For the MLD example, we choose $\bth_4=(\bmu,\bSi)=\left(\begin{array}{lll}
1 & 0.1 & 0.1 \\
1 & 0.1 & 0.2
\end{array}\right)$, $n \in \{50,100,500,1000,5000\}$.
For the MALD example, we choose $\bth_5=(\bmu, \bmu^*, \bSi)=\left(\begin{array}{llll}
0 & 1 & 0.1 & 0.1 \\
0 & 1 & 0.1 & 0.2
\end{array}\right)$, $n \in \{50,100,200,500\}$.
For the Type \II \ MLD example, we choose $\bth_6=(\bmu,\bSi)=\left(\begin{array}{lll}
0 & 0.1 & 0.1 \\
0 & 0.1 & 0.2
\end{array}\right)$, $n \in \{50,100,200\}$.
For the LMM example, we set the $\bbe=(2,1,5,1,4,12,3)^{\T}$, $\bpsi_{\th_1}=\text{diag}(0.5,0.5,0.25,0.25,0.25,0.5)$,  $\bLa_{\th_2}=\si^2\bI$ and $\si^2=0.5$. In this case, Z structure divides the data evenly into six groups, $\bth_7=\{\bbe, \bpsi_{\th_1},\si^2\}$. We choose $n \in\{18,180,360\}$. For the GPED example, we choose $\th=2$, $n \in \{20000,100000,500000,1000000\}$. For the FND example, we choose $\bth_8=(\mu, \si )=(3,0.5)^{\T}$, $n \in \{100,500,1000,5000,10000\}$.

For all initial values of parameters, we set them to 0 if they can be 0, and 1 otherwise. The precision $\de$ is set to be $10^{-6}$. The bootstrap method is adopted to obtain the MSE. Finally the above process is repeated 1,000 times to compute the coverage probability of the percentile bootstrap CI. The simulation results are summarized in Table 1, 2 and 3.

For the visualization of results, we consider both scenarios with STND and TTND case, where the sample size is set to be $n = 600(100)1500$, $a(s)b$ means from $a$ to $b$ with step size being $s$. For the case of STND, we choose the setting of $(\mu,\si,\de)$ as $(3,0.5,1)^{\T}$. The estimation performance and tendency including the MLEs of $(\mu,\si,\de)$ and the iteration times when the values of parameters converged under different sample sizes are described by the boxplots. For the case of TTND, we choose the setting of $(\mu,\si)$ as $(2,0.8)^{\T}$.

%

The simulation results show the convergence of the algorithm and that it is more stable and converges faster as the number of samples increases.

In those simulations, the estimated results of the parameters are also quite satisfactory. With the various sample sizes, it is relatively stable for the MSE and the CP of the percentile bootstrap CIs, where the CPs of the percentile bootstrap CIs are over 0.90 even 0.95 with excellent performance. it is delighted to discover that the MSE of the MLEs decreases with the increasing sample sizes and fluctuate steadily in a small range. In the meantime, with the number of parameters increasing, the consumption of the N-EM algorithms are adding since each iteration is computationally much more expensive. The CPs of the percentile bootstrap CIs for FND example are significantly lower than in other examples because the CPs of $\si^2$ are not performing well, the specific reason refers to Liu \et [26].
\section{$\!\!\!\!\!\!\!\!$. Convergence properties}  

We establish the theoretical properties of the proposed N-EM algorithms including local convergence and global convergence, and convergence rate. The proofs are also provided.

Let $\ell(\bth)$ be the function to be maximized and $Q(\bth|\btht)$ be the surrogate function, where $\bth$ is the parameter vector and $\btht$ be its current estimate. Denote the maximizer of $Q(\cdot|\bth)$ by $M(\bth)$. Following Proposition 15.3.1 and Proposition 15.4.3 of in Lange  (2010), we first give general and verifiable conditions for proving the local and global convergence of an N-EM sequence.

\begin{proposition}\label{p1}
Assume that the surrogate function $Q(\bth|\btht)$ is strictly concave, then the N-EM algorithm is locally attracted to a local optimum $\bth^\infty$ at a linear rate equal to the spectral radius of $I-d^{20}Q(\bth^\infty|\bth^\infty)^{-1}d^2\ell(\bth^\infty)$.
\end{proposition}

\begin{proof}
Since $\nabla Q(M(\bth)|\bth) = \bf 0$, it is easy to show that $M(\bth)$ is continuously differentiable with differential
\begin{equation}\label{eqn5.1}
    \rd M(\bth) = -\rd^{20}Q(M(\bth)|\bth)^{-1} \rd^{11}Q(M(\bth)|\bth).
\end{equation}
Furthermore, $\nabla\ell(\bth) - \nabla Q(M(\bth)|\bth) = \bf0$. Taking differential on both sides and set $\bth=\bth^\infty$, we have
\begin{equation}\label{eqn5.2}
    \rd^2\ell(\bth^\infty) - \rd^{20}Q(\bth^{\infty}|\bth^\infty) -\rd^{11}Q(\bth^\infty|\bth^\infty) = \bf0.
\end{equation}
Substituting (\ref{eqn5.2}) into (\ref{eqn5.1}), we have $\rd M(\bth^\infty)=I-\rd^{20}Q(\bth^\infty|\bth^\infty)^{-1}\rd^2\ell(\bth^{\infty})$.
By Lange's Lemma,
it is then sufficient to show that all the eigenvalues of the differential $\rd M(\bth^\infty)$ belong to $[0,1)$. Here we determine the eigenvalues of $\rd M(\bth^\infty)$ by the stationary values of the Rayleigh quotient
$$
    R(\ibv) = \frac{\ibv^{\T}[\rd^{20}Q(\bth^{\infty}|\bth^\infty)-\rd^2\ell(\bth^\infty)]\ibv}{\ibv^{\T}\rd^{20}Q(\bth^{\infty}|\bth^\infty)\ibv}=1-\frac{\ibv^{\T}\rd^2\ell(\bth^\infty)\ibv}{\ibv^{\T}\rd^{20}Q(\bth^{\infty}|\bth^\infty)\ibv}.
$$
At the optimal point $\bth^\infty$, both $\rd^2\ell(\bth^\infty)$ and $Q(\bth^\infty|\bth^\infty)$ are negative definite and $R(\ibv)<1$ for any unit vector $\ibv$.  The maximum of $R(\ibv)$ is strictly less than 1. Note also that $\rd^{20}Q(\bth^{\infty}|\bth^\infty)-\rd^2\ell(\bth^\infty) $ is negative semidefinite. It follows that $R(\ibv)\ge0$ and the minimum of $R(\ibv)$ is not less than $0$.
\end{proof}
The mapping functions $\bth^{(t+1)}=M(\btht)$ of the examples in Section 4 are differentiable and  the  surrogate functions in (\ref{eqn3.NEMTGL.3}), (\ref{eqn3.NEMTGL.8}), (\ref{eqn3.NEMTGL.13}), (\ref{eqn3.NEMTGL3.17}), (\ref{eqn3.NEMTGL3.23}), (\ref{eqn3.2.30}), (\ref{eqn3.2.34}) and (\ref{eqn3.2.36}) are strictly concave. The
local convergence results follow directly by Proposition 1.

\begin{corollary}\label{c1}
With an initial value $\bth^{(0)}$, the sequences $\{\btht\}$ generated by the N-EM algorithms which update the estimates by (\ref{eqn3.NEMTGL.4}), (\ref{eqn3.NEMTGL.9}), (\ref{eqn3.NEMTGL3.14}), (\ref{eqn3.NEMTGL3.18}), (\ref{eqn3.NEMTGL3.24}), (\ref{eqn3.2.31}), (\ref{eqn3.2.35}) and (\ref{eqn3.2.37}) respectively, are convergent to a local optimal $\bth^{\infty}$.
\end{corollary}
A function $f: \bbR^q \rightarrow \bbR \bigcup \{-\infty, \infty\}$ is coercive if and only if $f(x)\rightarrow \infty$ as $ ||x||_2\rightarrow \infty,$ where $||\cdot||_2$ denotes the standard Euclidean norm.
\begin{proposition}\label{p2}
Assume that $-\ell(\bth)$ is coercive, the subset $\{\bth \in \Om: \ell(\bth)\ge\ell(\btht)\}$ of parameter domain $\Om$ is compact and all stationary points of $\ell(\bth)$ are isolated.  The  minorizing function $Q(\bth|\btht)$ constructed by the N-EM algorithm is strictly concave and differentiable on both $\bth$ and $\btht$. Then the N-EM sequence $\bth^{(t+1)} = M(\btht)$  converges to the stationary point of $\ell(\bth)$. If $\ell(\bth)$ is strictly concave, then the limiting point of $\{\btht\}$ is the maximum.
\end{proposition}
\begin{proof}
Let $\Ga$ be the set of cluster points generated by the sequence $\bth^{(t+1)} = M(\btht)$ starting from the initial value $\bth^{(0)}$. By the Liapunov's theorem in  , $\Ga$ is contained in the set $\Delta$ of stationary points of $\ell(\bth)$. On the other hand, $\Ga$ is a closed subset of the compact set $\{\bth \in \Om: \ell(\bth)\ge\ell(\bth^{(0)})\}$ and this implies $\Ga$ is also compact. According to Proposition 8.2.1 in Lange (2010), $\Ga$ is connected. The condition that all stationary points of $\ell(\bth)$ are isolated easily implies that the number of stationary points in the compact set $\{\bth \in \bOm: \ell(\bth)\ge\ell(\bth^{(0)})\}$ can only be finite. Since the cluster set $\Ga$ is a connected subset of  finite set $\Delta$,  $\Ga$ reduces to a singleton.
\end{proof}
We thus have the following result for the examples in Section 4.

\begin{corollary}\label{c2}
Assume the differentiability and coerciveness of $-\ell(\bth)$ hold, all stationary points of $-\ell(\bth)$ are isolated and the subsets $\{\bth \in \bOm: \ell(\bth)\ge\ell(\btht)\}$ of parameter domain $\bOm$ are compact for examples in Section 5, 6 and 7. Then the sequence of iterates in (\ref{eqn3.NEMTGL.4}), (\ref{eqn3.NEMTGL.9}), (\ref{eqn3.NEMTGL3.14}), (\ref{eqn3.NEMTGL3.18}), (\ref{eqn3.NEMTGL3.24}), (\ref{eqn3.2.31}), (\ref{eqn3.2.35}) and (\ref{eqn3.2.37})  converge to the stationary points of $\ell(\bth)$ respectively. If the strict concavity of $\ell(\bth)$ also hold for the eight examples, then the sequence of iterates in (\ref{eqn3.NEMTGL.4}), (\ref{eqn3.NEMTGL.9}), (\ref{eqn3.NEMTGL3.14}), (\ref{eqn3.NEMTGL3.18}), (\ref{eqn3.NEMTGL3.24}), (\ref{eqn3.2.31}), (\ref{eqn3.2.35}) and (\ref{eqn3.2.37}) converge to the maximum points of $\ell(\bth)$ respectively.
\end{corollary}

\section{$\!\!\!\!\!\!\!\!$. Discussion}  

Motivated by MLEs of complex likelihood function which contains the integral terms, we establish a unified N-EM framework to deal with the above situation. What's more important is that the N-EM algorithm solves the most critical problem faced in the use and it is a promotion of EM algorithms. In the original EM algorithm usage scenario, the user needs to design the EM algorithm according to the model characteristics to find missing-data structure case by case, and this design requires the user to be quite familiar with the EM algorithm and the objective model, still needs some inspiration. In the proposed N-EM algorithm, the biggest advantage is that it only needs to analyze the structure of the objective function, find the integral term or homogeneous discrete sum, and then follow the algorithm steps we introduced to automatically construct the normalized density function and M-iterations. Therefore, this algorithm reduces the complexity and difficulty of using the algorithm, basically does not require the user's experience and professional background, and is formula-driven rather than inspiration-driven, thus expanding the use of the EM algorithm for more users and applied situations. What is more, the N-EM algorithm can solve some problems that cannot be solved using the EM algorithm, and we can group three classic kinds of problems into the N-EM framework. Another contribution of the N-EM algorithm is to solve the problem of finite discrete homogeneous sum, which can be regarded as the integral term. At last, because of the form of normalized density function, we can consider a N-EM interpretation of EM. In theoretical perspectives, this algorithm has been proved to possess the ascent property and the optimality of normalized density function, then the theoretical properties of the N--EM algorithms including the local and global convergence under some general assumptions are well studied. The applications and numerical experiments can demonstrate the superiority of N-EM algorithm.

\baselineskip 0.25in \vspace{0.5cm}
\section*{References}
\baselineskip 0.30in
\begin{description} \itemsep=-\parsep \itemindent=-1.3cm

\item[ ] Brailean, J.C., Little, D., Giger, M.L., Chen, C.T. and Sullivan, B.J. (1992). Application of the EM algorithm to radiographic images. \textit{Medical Physics} {\bf 19}(5), 1175--1182.

\item[ ] Cappe, O. (2011). Online EM algorithm for hidden Markov models. \textit{ Journal of Computational and Graphical Statistics} {\bf 20}(3), 728--749.

\item[ ] Casella, G. and George, E.I. (1992). Explaining the gibbs sampler. \textit{ The American Statistician} {\bf 46}(3), 167--174.

\item[ ] Celeux, G. (1985). The SEM algorithm, a probabilistic teacher algorithm derived from the EM algorithm for the mixture problem. \textit{Computational statistics quarterly} {\bf 2}, 73--82.

\item[ ] Celeux, G. and Diebolt, J. (1986). The SEM and EM algorithms for mixtures, Numerical and statistical aspects. \textit{Procedings of the 7th Franco Belgian Meeting of Statistics. Bruxelles, Publication Des Facultes Universitaries St. Louis}

\item[ ] Chen, L. and Jiang, Q. (2008). An extended EM algorithm for subspace clustering. \textit{Frontiers of Computer Science in China} {\bf 2}(1), 81--86.

  \item[ ] Cohen AC Jr. (1949). On estimating the mean and standard deviation of truncated normal distributions. \textit{Journal of the American Statistical Association} {\bf 44}, 518--525.

\item[ ] Dejardin, D. and Lesaffre, E. (2013). Stochastic EM algorithm for doubly interval-censored data. \textit{Biostatistics} {\bf 14}(4), 766--778.

\item[ ] Dempster, A.P., Laird N.M. and Rubin D.B. (1977). Maximum likelihood from incomplete data via the EM algorithm (with discussions). \textit{Journal of the Royal Statistical Society B} {\bf 39}, 1--38.

\item[ ] Deng, W., Mou, T. and Niu, N. (2019). Alternating EM algorithm for a bilinear
model in isoform quantification from RNA-seq data. \textit{Bioinformatics} {\bf 3}(3), 805--812.

\item[ ] Fernandez, J. (2010). Missing observations in spatial models, a spectral EM algorithm. \textit{Journal of Computational and Graphical Statistics} {\bf 19}(3), 684--701.

\item[ ] G\'{o}mez, H.W., Venegas, O. and Bolfarine, H. (2007). Skew-symmetric distributions generated by the distribution function of the normal distribution. \textit{Environmetrics} \textbf{18}, 395--407.

\item[ ] Gupta, A.K., Chang, F.C. and Huang, W.J. (2002). Some skew-symmetric models. \textit{Random Operators Stochastic Equations} \textbf{10}, 133--140.

\item[ ] Hoogerheide, L., Opschoor, A. and Van Dijk, H.K. (2012). A class of adaptive importance sampling weighted EM algorithms for efficient and robust posterior and predictive simulation. \textit{Journal of Econometrics} {\bf 171}(2), 101--120.

\item[ ] Hung, W.L. and Chang-Chien, S.J. (2017). Learning-based EM algorithm for normal-inverse Gaussian mixture model with application to extrasolar planets. \textit{Journal of Applied Statistics} {\bf 44}(6), 978--999.

\item[ ] Jose, C.A. and  Burzykowski, T. (2010). A version of the EM Algorithm for proportional hazard model with random effects. \textit{Biometrical Journal} {\bf 47}(6), 847--862.

\item[ ] Kadir, S.N., Goodman, D.F. and Harris, K.D. (2014). High-dimensional cluster analysis with the masked EM algorithm. \textit{Neural computation} {\bf 26}(11), 2379--2394.

\item[ ] Karlis, D. and Xekalaki, E. (1999). Improving the EM algorithm for mixtures. \textit{Statistics and Computing} {\bf 9}(4), 303--307.

\item[ ] Kay, J. (1997). The EM algorithm in medical imaging. \textit{Statistical Methods in Medical Research} {\bf 6}(1), 55--75.

\item[ ] Kim D.K. (1997). Regression analysis of interval-censored survival data with covariates using log-linear models. \textit{Biometrics } {\bf 53}(4), 1274--1283.

\item[ ] Lange, K. (2010). \textit{Numerical Analysis for Statisticians, Second Edition.}  Springer,  New York.

\item[ ] Lange, K., Hunter, D.R. and Yang, I. (2000). Optimization transfer using surrogate objective functions (with discussions). \textit{ Journal of Computational and Graphical Statistics} {\bf 9}, 1--20.

  \item[ ] Liu, P.Y., Tian, G.L., Yuen, K.C., Zhang, C. and Tang, M.L. (2021). Proportional inverse Gaussian distribution: A new tool for analyzing continuous proportional data.  \textit{Australian \& New Zealand Journal of Statistics}, in press.

  \item[ ] Liu, X., Tian, G.L., Fei, Y., Shu, L. and Zhao, Q. (2020). Folded normal regression models with applications in biomedicine. \textit{Journal of Computational and Applied Mathematics} {\bf 379}, 112941.

\item[ ] Lucas, A. (1997). Robustness of the student $t$ based M-estimator. \textit{Communications in Statistics --- Theory and Methods} \textbf{26}(5),  1165--1182.

\item[ ] McLachlan, G.J. and Krishnan, T. (2007). \textit{The EM algorithm and extensions.}  John Wiley \& Sons, Hoboken, New Jersey.

\item[ ] Nadarajah, S. and Kotz, S. (2003). Skewed distributions generated by the normal kernel. \textit{Statistics \& Probability Letters} \textbf{65}, 269--277.

\item[ ] Nakashima, S., Sughiyama, Y. and Kobayashi T.J. (2020). Lineage EM algorithm for inferring latent states from cellular lineage trees. \textit{ Bioinformatics} {\bf 36}(9), 2829--2838.

\item[ ] Panic, B., Klemenc, J. and Nagode, M. (2020). Improved initialization of the EM algorithm for mixture model parameter estimation. \textit{ Mathematics} {\bf 8}(3), 373.

\item[ ] Ranjan, R., Huang, B. and Fatehi, A. (2016). Robust Gaussian process modeling using EM algorithm. \textit{Journal of Process Control} {\bf 42}, 125-136.

  \item[ ] Ripley, B.D. (2009). \textit{Stochastic simulation.} John Wiley \& Sons, New York.

  \item[ ] Rubin, D.B. (1987). The calculation of posterior distributions by data augmentation, Comment, A non-iterative sampling/importance resampling alternative to the data augmentation algorithm for creating a few imputations when fractions of missing information are modest, The SIR algorithm. \textit{Journal of the American Statistical Association} {\bf 82}(398), 543--546.

\item[ ] Suesse, T. and Zammit-Mangion, A. (2017). Computational aspects of the EM algorithm for spatial econometric models with missing data. \textit{Journal of Statistical Computation and Simulation} {\bf 87}(9), 1767--1786.

\item[ ] Teimouri, M. (2020). EM algorithm for mixture of skew-normal distributions fitted to grouped data. \textit{Journal of Applied Statistics} \textbf{48}(7), 1--26.

\item[ ] Tian, G.L., Ding, X., Liu, Y. and Tang, M.L. (2019). Some new statistical methods for a class of zero-truncated discrete distributions with applications. \textit{Computational Statistics} {\bf 34}(3), 1393--1426.

\item[ ] Tian, G. L., Huang, X.F. and Xu, J. (2019). An assembly and decomposition approach for constructing separable minorizing functions in a class of MM algorithms. \textit{Statistica Sinica} \textbf{29}(2), 961--982.

\item[ ] Tsagris, M., Beneki, C. and Hassani, H. (2014). On the folded normal distribution. \textit{Mathematics} {\bf 2}(2), 12--28.

  \item[ ] Tweedie, M.C.K. (1957). Statistical properties of inverse Gaussian distributions. I. \textit{The Annals of Mathematical Statistics} \textbf{28}(2), 362--377.

\item[ ] Vaida, F. and Liu, L. (2009). Fast implementation for normal mixed effects models with censored response. \textit{Journal of Computational and Graphical Statistics} {\bf 18}(4), 797--817.

\item[ ] Wang, X., Schumitzky, A. and David, Z.D. (2007). Nonlinear random effects mixture models, Maximum likelihood estimation via the EM algorithm. \textit{Computational Statistics \& Data Analysis} {\bf 51}(12), 6614--6623.

\item[ ] Wei, G.C. and Tanner, M.A. (1990). A monte carlo implementation of the EM algorithm and the poor man's data augmentation algorithms. \textit{Journal of the American statistical Association} {\bf 85}(411), 699704.

\item[ ] Wei, G.C. and Tanner, M.A. (1990a). Posterior computations for censored regression data. \textit{Journal of the American Statistical Association} {\bf 85}(411),829--839.

\item[ ] Wu, D. and Ma, J. (2018). A two-layer mixture model of Gaussian process functional regressions and its MCMC EM algorithm. \textit{IEEE transactions on neural networks and learning systems} {\bf 29}(10), 4894--4904.

\item[ ] Wu, D. and Ma, J. (2019). An effective EM algorithm for mixtures of Gaussian processes via the MCMC sampling and approximation. \textit{ Neurocomputing} {\bf 331}, 366--374.

\item[ ] Xie, M. and Simpson, D. (2015). Regression modeling of ordinal data with nonzero baselines. \textit{Biometrics} {\bf 55}(1), 308--316.

\item[ ] Xu, L. and Jordan, M.I. (1993). Unsupervised learning by EM algorithm based on finite mixture of Gaussians. \textit{1993 World Congress on Neural Networks} {\bf 2}, 431--434.

\item[ ] Yang, F. (2018). A stochastic EM algorithm for quantile and censored quantile regression models. \textit{Computational Economics} {\bf 52}(2), 555--582.

\end{description}

\newpage
\vkU \baselineskip 0.10in
\begin{Table}
\centering
\label{tab1} \ {\rm Simulation results for examples in Section 4.1
\vspace{-0.2cm}
\begin{center}
\renewcommand{\arraystretch}{1.10} \tabcolsep 0.09in \doublerulesep 1.0pt
\begin{tabular}{lccccccccccc} \hline \hline
  Settings    & $n$ &  No. of Para.   &  $K$       && Time    &&   $L$   & &   MSE & & CP  \\ \hline
\MC{10}{c}{SNND} \\
$\bth_1$    & 1000  &  3      &  360    &&   0.27   &&  $-1201.33 $ && 1.68e-02 && 0.9459 \\
$\bth_1$    & 10000 &  3      & 349     &&  1.67    &&  $-13190.57$ && 2.28e-03 && 0.9768 \\
$\bth_1$    & 50000 &  3      &  346    &&    8.10  &&  $-69910.33$ && 4.94e-04 &&  0.9162 \\ \hline
\MC{10}{c}{STND} \\
$\bth_2 $    & 1000  &  3      &  376    &&   0.26   &&  $-1703.30 $ && 3.01e-02 && 0.9684 \\
$\bth_2 $    & 10000 &  3      & 373     &&  1.98    &&  $-18054.59$ && 3.61e-03 && 0.9584 \\
$\bth_2 $    & 50000 &  3      &  372    &&    9.65  &&  $-98280.72$ && 5.88e-04 && 0.9056 \\ \hline
\MC{10}{c}{TTND} \\
$\bth_3$    & 500   &  2      &  5    &&   1.49e-05  &&  $-5354.36$  && 3.85e-04 && 0.9358 \\
$\bth_3$    & 1000  &  2      & 5     &&  1.31e-04   &&  $-11240.57$ && 1.78e-04 && 0.9641 \\
$\bth_3$    & 5000  &  2      &  5    &&    3.09e-04 &&  $-59876.49$ && 4.05e-05 && 0.9598 \\
$\bth_3$    & 10000 &  2      & 5     &&  4.97e-04   &&  $-13965.50$ && 1.97e-05 && 0.9781 \\
$\bth_3$    & 50000 &  2      & 5     &&  1.16e-03   &&  $-72652.00$ && 3.62e-06 && 0.9148 \\
\hline
\end{tabular}
\end{center}
}\end{Table}
\baselineskip 0.10in \vspace{-0.2cm}
\noi {\small K: iteration numbers; Time: run times (in seconds); L:  final objective values; \\
MSE $= \frac{1}{R}\sum_{r=1}^R \frac{||\hat{\bth}-\bth||^2}{q}$, where $q$ indicates the number of parameters. \\
CP: the coverage proportion of the percentile bootstrap CIs.}

\newpage
\vkU \baselineskip 0.10in
\begin{Table}
\centering
\label{tab2} \ {\rm Simulation results for examples in Section 4.2
\vspace{-0.2cm}
\begin{center}
\renewcommand{\arraystretch}{1.10} \tabcolsep 0.075in \doublerulesep 1.0pt
\begin{tabular}{lccccccccccc} \hline \hline
Settings    & $n$ &  No. of Para.   &  $K$       && Time    &&  $L$   & &   MSE & & CP  \\ \hline
\MC{10}{c}{MLD} \\
  $\bth_4$    & 50   &  6      &  9    &&   0.04  && $-117.89$   && 7.51e-03 && 0.9684 \\
  $\bth_4$    & 100  &  6      & 8     &&  0.07   && $-244.72$   && 2.12e-03 && 0.9791 \\
  $\bth_4$    & 500  &  6      &  6    &&    0.26 && $-1388.01$  && 2.24e-04 && 0.9484 \\
  $\bth_4$    & 1000 &  6      & 5     &&  0.44   && $-2860.00$  && 9.10e-05 && 0.9198 \\
  $\bth_4$    & 5000 &  6      & 5     &&  1.53   && $-14312.25$ && 1.84e-05 && 0.9489 \\ \hline
\MC{10}{c}{MALD} \\
$\bth_5$    & 50  &  8      & 20     &&   31.44  && $-122.09$  && 0.13 && 0.9634 \\
$\bth_5$    & 100 &  8      & 18     &&  55.95   && $-257.16$  && 0.06 && 0.9168 \\
$\bth_5$    & 200 &  8      & 15     && 88.86    && $-521.51$  && 0.04 && 0.9813 \\
$\bth_5$    & 500 &  8      & 11     &&  187.37  && $-1420.03$ && 0.01 && 0.9185 \\ \hline
\MC{10}{c}{Tpye II MLD} \\
$\bth_6$    & 50  &  6      & 12    &&  206.00   && $-113.07$  && 1.63e-03 && 0.9789 \\
$\bth_6$    & 100 &  6      & 11    &&  412.50   && $-241.03$  && 9.14e-04 && 0.9219 \\
$\bth_6$    & 200 &  6      &  8    &&   731.72  && $-496.12$  && 4.02e-04 && 0.9414 \\ \hline
\MC{10}{c}{LMM} \\
$\bth_7$    & 18  &  44        &  239     &&  0.09   && $-172.11$  && 0.14 && 0.9165 \\
$\bth_7$    & 180 &  44        &  198     &&  0.36   && $-1871.66$ && 0.03 && 0.8916 \\
$\bth_7$    & 360 &  44        &  179     &&  1.37   && $-3829.56$ && 0.02 && 0.9416 \\ \hline
\MC{10}{c}{GPED} \\
$\th $    & 20000   &  1      & 5       &&   1.99e-04 && $-43043.91$   && 1.28e-03 && 0.9486 \\
$\th $    & 100000  &  1      & 5       && 2.99e-04   && $-223123.27 $ && 1.08e-03 && 0.9243 \\
$\th $    & 500000  &  1      & 5       && 5.99e-04   && $-1151116.20$ && 1.00e-03 && 0.9974 \\
$\th $    & 1000000 &  1      & 5       && 1.80e-03   && $-2801671.68$ && 7.94e-04 && 0.9431 \\
\hline
\end{tabular}
\end{center}
}\end{Table}

\vkU \baselineskip 0.10in
\begin{Table}
\centering
\label{tab3} \ {\rm Simulation results for examples in Section 4.3
\vspace{-0.2cm}
\begin{center}
\renewcommand{\arraystretch}{1.10} \tabcolsep 0.095in \doublerulesep 1.0pt
\begin{tabular}{lccccccccccc} \hline \hline
Settings    & $n$ &  No. of Para.   &  $K$     && Time    &&   $L$   & &   MSE & & CP \\ \hline
\MC{10}{c}{FND} \\
$\bth_8$    & 100   &  2      & 15    && 0.24  &&  $-337.00$   && 0.05     && 0.8642 \\
$\bth_8$    & 500   &  2      & 14    && 0.25  &&  $-1754.83$  && 7.59e-03 && 0.9067 \\
$\bth_8$    & 1000  &  2      & 14    && 0.25  &&  $-3699.15$  && 4.10e-03 && 0.8733 \\
$\bth_8$    & 5000  &  2      & 14    && 0.39  &&  $-20663.56$ && 7.23e-04 && 0.8239 \\
$\bth_8$    & 10000 &  2      & 11    && 0.48  &&  $-43073.68$ && 5.71e-04 && 0.8416 \\
\hline
\end{tabular}
\end{center}
}\end{Table}

\end{document}